\documentclass[11pt,a4paper]{amsart}

\title[Critical temperatures and collapsing of two-dimensional Log gases]{Critical temperatures and collapsing of two-dimensional Log gases}

\author{Rolf Andreasson \& Ludvig Svensson}

\address{Rolf Andreasson, Department of Mathematical Sciences, University of Gothenburg and 
Chalmers University of Technology, SE-412 96 G\"{o}teborg, Sweden}
\email{rolfan@chalmers.se}

\address{Ludvig Svensson, Department of Mathematical Sciences, University of Gothenburg and 
Chalmers University of Technology, SE-412 96 G\"{o}teborg, Sweden}
\email{ludsven@chalmers.se}

\date{\today}

\usepackage[english]{babel}
\usepackage{amsmath,calligra}
\usepackage{amsthm,amssymb,latexsym}
\usepackage[all,cmtip]{xy}
\usepackage{url,ae}
\usepackage{t1enc}
\usepackage{bm}
\usepackage{mathrsfs}
\usepackage{graphicx}
\usepackage{geometry}
\usepackage{comment}

\usepackage{adjustbox}

\usepackage{enumitem}
\usepackage{tabularx}
\usepackage{mathtools}
\usepackage[nameinlink]{cleveref}
\usepackage{xcolor}  

\crefname{figure}{Figure}{Figures}
\crefname{equation}{}{}
\crefname{table}{Table}{Tables}
\crefname{section}{Section}{Sections}
\crefname{appendix}{Appendix}{Appendices}
\crefname{corollary}{corollary}{corollaries}
\crefname{proposition}{proposition}{propositions}
\crefname{theorem}{theorem}{theorems}
\crefname{lemma}{lemma}{lemmas}
\crefname{definition}{definition}{definition}
\crefname{problem}{problem}{problems}
\crefname{example}{example}{examples}

\newtheorem{proposition}{Proposition}[section]
\newtheorem{theorem}[proposition]{Theorem}
\newtheorem{lemma}[proposition]{Lemma}
\newtheorem{corollary}[proposition]{Corollary}

\theoremstyle{definition}

\newtheorem{example}[proposition]{Example}
\newtheorem{remark}[proposition]{Remark}

\newtheorem{problem}[proposition]{Problem}

\numberwithin{equation}{section}

\newcommand{\C}{\mathbb{C}}

\renewcommand{\d}{\mathrm{d}}

\def\newop#1{\expandafter\def\csname #1\endcsname{\mathop{\rm #1}\nolimits}}

\begin{document}
\nocite{*}
\bibliographystyle{plain}

\begin{abstract}
    We consider the canonical ensemble of a system of point particles on the sphere interacting via a logarithmic pair potential. In this setting, we study the associated Gibbs measure and partition function, and we derive explicit formulas relating the critical temperature, at which the partition function diverges, to a certain discrete optimization problem. We further show that the asymptotic behavior of both the partition function and the Gibbs measure near the critical temperature is governed by the same optimization problem. Our approach relies on the Fulton--MacPherson compactification of configuration spaces and analytic continuation of complex powers. To illustrate the results, we apply them to well-studied systems, including the two-component plasma and the Onsager model of turbulence. In particular, for the two-component plasma with general charges, we describe the formation of dipoles close to the critical temperature, which we determine explicitly. 
\end{abstract}

\maketitle
\thispagestyle{empty}

\section{Introduction}
Consider the following interaction energy for $N(\geq 2)$ point particles at positions $p_1,\dots,p_N$ on the unit sphere $\mathbb{S}^2$:
\begin{equation}
    \label{eq:energy the sphere}
    E(p_1,\dots,p_N)= -\sum_{i<j} c(i,j) \log d(p_i,p_j)^{2},
\end{equation}
where $d(\cdot,\cdot)$ denotes the chordal distance on $\mathbb{S}^2$ and $c(i,j)$, for $1\leq i<j\leq N$, are real numbers specifying the coupling between particles $i$ and $j$. The prototypical example is $c(i,j)=k_ik_j$, where $k_i$ represents the charge of particle~$i$. 

We consider the corresponding canonical ensemble at inverse temperature~$\beta$. The associated Gibbs measure on $(\mathbb{S}^2)^N$, when it exists, is given by
\begin{equation}
    \label{eq: gibbs measure}
    \begin{aligned}
        \mu_{\beta} &= \frac{1}{\mathcal{Z}(\beta)}\exp(-\beta (E(p_1,\dots,p_N))\,\d V^{\otimes N} \\
        &= \frac{1}{\mathcal{Z}(\beta)}\prod_{i< j} d(p_i,p_j)^{2c(i,j)\beta} \d V^{\otimes N}.
    \end{aligned}
\end{equation}
This measure exists for all values of $\beta$ such that the \textit{partition function}
\begin{equation}
    \label{eq: partition function intro}
    \mathcal{Z}(\beta) = \int_{(\mathbb{S}^2)^N}\prod_{i< j} d(p_i,p_j)^{2c(i,j)\beta}\d V^{\otimes N}
\end{equation}
is finite. Here $\mathrm{d}V$ is the symmetric, normalized volume form on the sphere. In general, there are \textit{critical values} of $\beta$ at which $\mathcal{Z}(\beta)$ becomes infinite, causing the Gibbs measure to cease to exist. This phenomenon is highly dependent on the logarithmic singularity of the pair-interaction. In contrast, for Coulomb gases in other dimensions, with singular pair-interaction governed by the corresponding Green's function to the Laplace equation, no similar behavior occurs; the Gibbs measure either exists for all temperatures or for none. 

This model encompasses several important and well-studied particle systems in physics, including one- and two-component (and, more generally, $n$-component) plasmas on the sphere. Such two-dimensional plasmas, in various geometries, have been proposed as simplified models for dimensionally reduced, translation-invariant three-dimensional plasmas, see \cite{Ki,Ta}. Furthermore, the same framework also describes systems of quasi-particles such as vortices. Vortices described by a Log gas arise in several different contexts, including the Ginzburg--Landau model of superconductivity, see \cite{SaSe}, the analysis of the Berezinskii--Kosterlitz--Thouless transition for the XY model, see \cite{KoTh}, and the Onsager model of turbulence for two-dimensional Eulerian fluids, see, e.g., \cite{O}. In the latter case, as well as in related quasi-particles systems such as the guiding center plasma, see \cite{ET}, both positive and negative values of the inverse temperature are of interest. Hence, we will treat $\beta$ as an arbitrary real parameter. 

\medskip

Before presenting our results, we introduce two discrete optimization problems that are central to all of them. 
\begin{problem}
    \label{prob: max}
    Given an integer $N$ and real numbers $c(i,j), 1\leq i<j\leq N$, consider the minimization problem
    \begin{equation*}
        T^+\coloneq-\min_S \frac{\sum_{i,j \in S:i<j} c(i,j)}{|S|-1},
    \end{equation*}
    where $S$ ranges over all subsets of $\{1,\dots,N\}$ of size at least two. 
\end{problem}
\begin{problem}
    \label{prob: min}
    Given an integer $N$ and real numbers $c(i,j), 1\leq i<j\leq N$, consider the maximization problem
    \begin{equation*}
        T^-\coloneq-\max_S \frac{\sum_{i,j \in S:i<j} c(i,j)}{|S|-1},
    \end{equation*}
    where $S$ ranges over all subsets of $\{1,\dots,N\}$ of size at least two. 
\end{problem}
Variants of these problems have appeared in the literature in other contexts, as discussed in \Cref{sec: other problems}. Our first result is an explicit formula for the critical inverse temperatures, given in terms of the optimization problems above.
\begin{theorem}
    \label{thm:1}
    The partition function $\mathcal{Z}(\beta)$ is finite precisely on the interval $(\beta^-,\beta^+)$ where 
    \begin{equation}
        \label{eq: crit inv temp minus}
        \beta^-=\begin{cases}
            1/T^-\text{ if }T^-<0\\
            -\infty \text{ otherwise}
        \end{cases}
    \end{equation}
    and
    \begin{equation*}
        \label{eq: crit inv temp plus}
        \beta^+=
        \begin{cases}
            1/T^+\text{ if }T^+>0\\
            \infty \text{ otherwise.}
        \end{cases}
    \end{equation*}
\end{theorem}
When the couplings $c(i,j)$ are all positive, \cref{eq: crit inv temp minus} is an explicit formula for the \emph{log canonical threshold}, an important singularity invariant in algebraic geometry, of a certain effective divisor defined by $c(i,j)$, see \cref{eq:Dc} below. In this case, \Cref{thm:1} resembles results in \cite{Te} and \cite{Mu}, and our proof roughly follows the proof in \cite{Te}. 

Solving \Cref{prob: max,prob: min} can be quite challenging in practice. Nevertheless, \Cref{thm:1} provides a direct way to find explicit upper bounds for $\beta^+$ and lower bounds for $\beta^-$ by simply evaluating the expression for any chosen subset $S$. By relating these optimization problems to Rayleigh--Ritz quotients and the min-max theorem for eigenvalues of symmetric matrices, we show in \Cref{prop: bounds in terms of eigenvalues} how to obtain bounds in the opposite, more difficult direction. The bounds are in terms of eigenvalues of the symmetric matrix with entries $c(i,j)$. This leads to stochastic bounds in a quenched, random version of the Log gas, considered, for example, in \cite{KiWa}, where the couplings are assumed to be random variables. We consider both the case of independent couplings, see \Cref{prop: independent random coupling}, and that of independent charges, see \Cref{prop: independent random charges}, which exhibit markedly different behavior. 

The eigenvalue bounds also leads to the following explicit bounds in the key case when $c(i,j)=k_ik_j$, $1\leq i<j\leq N$, for real numbers $k_i$. 
\begin{corollary}
    \label{cor: explicit bounds in the charge case}
    Let $c(i,j)=k_ik_j,1\leq i<j\leq N$ for real numbers $k_i,1\leq i\leq N$. If $\beta^+$ is finite, then
    \begin{equation}
        \label{eq: bound on beta plus for charges}
        \beta^+ \geq \frac{1}{\max_{1\leq i \leq N}k_i^2}.
    \end{equation}
    If $\beta^-$ is finite, then
    \begin{equation}
        \label{eq: bound on beta minus for charges}
        \beta^- \leq  -\frac{1}{\big(\sum_{i=1}^N k_i^2\big)-\min_{1\leq i \leq N} k_i^2}.
    \end{equation}
\end{corollary}
Our second result concerns the limiting behavior of the Gibbs measure $\mu_\beta$ as $\beta$ approaches the finite endpoints of the interval $(\beta^-,\beta^+)$. To state the result, we first introduce some additional notation. Let $G^+$ and $G^-$ denote the sets of subsets of $\{1,\dots, N\}$ that realize the minimum in \Cref{prob: max} and the maximum in \Cref{prob: min}, respectively. That is, 
\begin{equation}
    G^+= \Big\{S\subseteq \{1,\dots,N\}: |S|\geq 2 \text{ and }-\frac{\sum_{i<j\in S}c(i,j)}{|S|-1} = T^+ \Big\},
\end{equation}
and similarly for $G^-$. Following \cite{FM}, we say that a collection $K =\{S_1,\dots,S_k\}$ of distinct subsets of $\{1,\dots,N\}$ is a \emph{nest} if any two elements $S_i,S_j \in K$ are \emph{nested}, meaning that $S_i\subset S_j$, $S_j\subset S_i$, or $S_i\cap S_j=\emptyset$. Let $\mathcal{N}^{+}$ and $\mathcal{N}^-$ denote the sets of nests of $G^+$ and $G^-$, respectively, having maximal size $\kappa^+$ and $\kappa^-$ among all nests of $G^+$ and $G^-$. 
\begin{theorem}
    \label{thm:2}
    As $\beta$ approaches one of the endpoints of the interval $(\beta^-,\beta^+)$, if the endpoint is finite, we have the following weak limits:
    \begin{equation}
        \mu_\beta \xrightarrow[\beta \to \beta^+]{} \mu^+
    \end{equation}
    and
    \begin{equation}\label{eq: weak limit of gibbs measures}
        \mu_\beta \xrightarrow[\beta \to \beta^-]{} \mu^-.
    \end{equation}
    %
    Here, $\mu^+$ is a probability measure whose support is
    \begin{equation}
        \label{eq: support set +}
        \bigcup_{K\in \mathcal{N}^+} \bigcap_{S\in K}  \big\{(p_1,\dots,p_N)\in (\mathbb{S}^2)^N\colon p_{j_1}=\cdots=p_{j_k},\, S=\{j_1,\dots,j_k\}\big\},
    \end{equation}
    and $\mu^-$ is a probability measure whose support is
    \begin{equation}
        \label{eq: support set -}
        \bigcup_{K\in \mathcal{N}^-} \bigcap_{S\in K}  \big\{(p_1,\dots,p_N)\in (\mathbb{S}^2)^N\colon p_{j_1}=\cdots=p_{j_k},\, S=\{j_1,\dots,j_k\} \big\}.
    \end{equation}
\end{theorem}
The weak convergence in \Cref{thm:2} and the characterization of the support of the limit implies that if $U$ is an open set compactly contained in the complement of \eqref{eq: support set +} in $(\mathbb{S}^2)^N$, then the probability that the system occupies $U$ converges to zero as $\beta\to \beta^+$, and similarly as $\beta$ approaches $\beta^-$. Thus, \Cref{thm:2} describes the collapse (or condensation) of the system as a critical temperature is approached, and \eqref{eq: support set +} and \eqref{eq: support set -} specify the possible clusters formed in the limit. We find examples where pairs are formed (see \cref{eq: support for two-comp plasma} below) and other examples displaying total collapse of all particles into a single cluster (see \Cref{rmk:gibbs} below).

A fundamental quantity in equilibrium statistical mechanics is the \emph{Gibbs free energy} as a function of the inverse temperature $\beta$, given in our setting by $-1/N\log{\mathcal{Z}(\beta)}$. Our third result concerns the asymptotic behavior of the free energy as $\beta$ approaches the critical values $\beta^+$ and $\beta^-$. Recall that $\kappa^+$ and $\kappa^-$ denote the sizes of the maximal nests in $G^+$ and $G^-$, respectively. 
\begin{theorem}
    \label{thm:3}
    The free energy admits the asymptotic behavior
    \begin{equation}
        -\frac{1}{N}\log\mathcal{Z}(\beta) = \frac{\kappa^\pm}{N}\log(\beta-\beta^\pm) + \mathcal{O}(1)\quad \text{as }\beta\to \beta^\pm,
    \end{equation}
    provided that $\beta^\pm$ is finite.
\end{theorem}
In fact, the proof provides a full asymptotic expansion of the free energy as $\beta$ approaches $\beta^\pm$:
\begin{equation}
    -\frac{1}{N}\log\mathcal{Z}(\beta)=\frac{\kappa^\pm}{N}\log(\beta-\beta^{\pm})+C_0^\pm+C_1^\pm(\beta-\beta^\pm)+C^\pm_2(\beta-\beta^{\pm})^2 + \cdots.
\end{equation}
This expansion follows from the fact that $\mathcal{Z}(\beta)$ admits a Laurent series expansion about $\beta^\pm$, established via analytic continuation. Moreover, the order of the pole of this Laurent series expansion is precisely $\kappa^\pm$. The coefficients $C_0^\pm, C_1^\pm,\dots$ can be expressed in terms of the Laurent series coefficients; for instance, $C_0^\pm=(\log A_0^\pm)/N$, where $A_0^\pm$ is the coefficient of $(\beta-\beta^\pm)^{-\kappa^\pm}$ in the expansion of $\mathcal{Z}(\beta)$. 

In particular, if \Cref{prob: max} and/or \Cref{prob: min} in \Cref{thm:1} has a unique solution, then $\kappa^+$ and/or $\kappa^-$ is equal to $1$, respectively. We also find examples where the solution of \Cref{prob: max} is highly degenerate, and $\kappa^+$ grows linearly with $N$, see \Cref{sec: positive temperature} below.
\subsection{Comment on the proofs}

The proofs of \Cref{thm:1,thm:2,thm:3} rely on the identification of $(\mathbb{S}^2)^N$ with the complex algebraic variety $(\mathbb{CP}^1)^N$. The singularity of the Gibbs measure corresponds to a well-studied singular subvariety of $(\mathbb{CP}^1)^N$. To prove \Cref{thm:1}, we use the explicit embedded resolution of singularities of this subvariety, provided by the Fulton--MacPherson compactification of the ordered configuration space of $\mathbb{CP}^1$ \cite{FM}. The proofs of \Cref{thm:2,thm:3} exploit the distribution-valued meromorphic continuation of complex powers \cite{Atiyah, BG}. 
\subsection{Applications}
\label{sec: applications}

We now highlight \Cref{thm:1,thm:2,thm:3} with some applications to well-studied systems of physical origin. 
\subsubsection{Positive temperature}
\label{sec: positive temperature}
First, consider the general two-component plasma, introduced in \cite{HV}, for which
\begin{equation}
    \label{eq: general planar two-component plasma}
    c(i,j) = \begin{cases}
        Z_1^2, 1\leq i,j \leq N_1\\
        Z_2^2, N_1< i,j \leq N\\
        -Z_1Z_2, 1\leq i \leq N_1, N_1<j\leq N \\
        -Z_1Z_2, N_1< i \leq N, 1\leq j\leq N_1.
    \end{cases}
\end{equation}
In this system there are two types of particles: $N_1$ particles of the first and $N_2=N-N_1$ of the second. The first type of particle has charge $Z_1$ and the second type has charge $-Z_2$, where $Z_1$ and $Z_2$ are positive real numbers. Without loss of generality, we assume $Z_1\geq Z_2$. For simplicity (though this is not crucial), we also assume overall charge neutrality, that is, $Z_1N_1=Z_2N_2$. In this case we are able to solve \Cref{prob: max} explicitly, and as a consequence of Theorems \ref{thm:1}, \ref{thm:2} and \ref{thm:3}, we find the following.

\begin{theorem}
    \label{thm: two-component log gas}
    With $c(i,j)$ as in \cref{eq: general planar two-component plasma}, we have
    \begin{equation}
    \label{eq: crit temp for two-comp plasma}
        \beta^+=1/(Z_1 Z_2),
    \end{equation}
    and $\mu^+$ is supported on
    \begin{equation}
    \label{eq: support for two-comp plasma}
        \bigcup_{\substack{I\subset \{N_1+1,\hdots,N\} \\ |I| = N_1}} \bigcap_{j=1}^{N_1} \{p_j=p_{I_j}\}
    \end{equation}
    where the union is over all ordered tuples $I \subset \{N_1+1,\hdots,N\}$ of size $N_1$. In addition, $\kappa^+=N_1$, so that,
    \begin{equation}
        -1/N\log(\mathcal{Z}(\beta)) = \frac{Z_2}{Z_1+Z_2}\log\bigg(\beta-\frac{1}{Z_1Z_2}\bigg) + \mathcal{O}(1).
    \end{equation}
\end{theorem}

The support \eqref{eq: support for two-comp plasma} can be understood as the union over all possible ways to pair up positive and negative particles into a maximal number of pairs, of the corresponding subsets describing these configurations. In light of \Cref{thm:2}, \Cref{thm: two-component log gas} shows rigorously the formation of dipoles as the critical inverse temperature $\beta^+=1/(Z_1Z_2)$ is approached.

In the physics literature, see, e.g., \cite{HV}, the critical inverse temperature $\Gamma=2/(Z_1Z_2)$ can be found, which is consistent with \eqref{eq: crit temp for two-comp plasma} due to a difference in normalization. To the authors' knowledge, however, this formula has not yet been rigorously proven, except when $Z_1=Z_2$ \cite{GuPa}. 

\medskip

In the recent work \cite{BoSe}, dipole formation was established in the case $Z_1=Z_2=1$ for $\beta\geq \beta^+ = 1/(Z_1 Z_2) = 1$, in the limit $\lambda\to 0$, where $\lambda$ denotes a regularization parameter. Moreover, a detailed asymptotic expansion of the free energy was obtained. It would be of interest to derive an even more detailed asymptotic expansion for both the free energy and Gibbs measure for a suitably regularized version of the system when $\beta\geq 1$. The physics literature suggests that such an expansion should include additional terms corresponding to the formation of neutral \emph{multipoles}, beyond the leading order dipole contribution. As $\beta$ is increased, terms corresponding to increasingly larger neutral multipoles are expected to appear. In fact, in the very recent work \cite{BoSe2}, such higher-order terms have been shown to appear. It is natural to ask whether \Cref{prob: max} is also related to these further transitions. We will address this question in a sequel to the present work. 

\subsubsection{Negative temperature}

Now, consider the coupling defined by
\begin{equation*}
    \label{eq: onsager model 1}
    c(i,j) = k_ik_j,
\end{equation*}
where the $k_i$ are arbitrary nonzero real numbers. This model generalizes the two-component plasma and arises, for example, as the statistical mechanical description of the Onsager model of turbulence, with the $k_i$ representing the vorticities of the point-vortex particles. Determining the positive critical inverse temperature explicitly seems difficult in general, although \Cref{thm:1} and \Cref{cor: explicit bounds in the charge case} provide explicit upper and lower bounds. 

Here, we instead focus on the negative-temperature regime, where imposing a strong condition on the variation among the vorticities allows us to solve \Cref{prob: min} explicitly, see \Cref{cor: onsager model}. More precisely, the condition is that 
\[
    \max_{i:k_i>0} k_i < \frac{3}{2}\min_{i:k_i>0} k_i,
\]
and
\[
    \max_{i:k_i<0} |k_i| < \frac{3}{2}\min_{i:k_i<0} |k_i|.
\]
In this case, we observe a fundamentally different qualitative behavior of the condensation as the critical temperature is approached. Unlike the positive-temperature case (for the two-component plasma), where many small clusters are formed, here one observes a total collapse of either all positively or all negatively charged particles, depending on a simple criterion. See \Cref{cor: onsager model} below for details.
\subsection{Other geometries}

We have stated our result in the setting of the sphere for simplicity, but we expect them to hold almost verbatim in a more general setting, with similar proofs. For any Riemann surface $M$, compact or otherwise, with a metric $g$ and a volume form $\mathrm{d}V$, one can consider an analogous energy $E\colon M^N\to \mathbb{R}$ constructed from pair-interactions governed by the corresponding Greens function $G$, defined with respect to $g$, where suitable boundary conditions are imposed if $M$ has a boundary. Specifically, the energy is given by
\begin{equation*}
    E(p_1,\dots,p_N) = \sum_{i<j}c(i,j)\log G(p_i,p_j).
\end{equation*}
Furthermore, one can allow the inclusion of a smooth external potential $U\colon M \to \mathbb{R}$, and define the Gibbs measure as 
\begin{equation*}
    \mu_\beta = \frac{1}{\mathcal{Z}(\beta)}\exp\Big(-\beta\big(E(p_1,\dots,p_N)+\sum_i U(p_i)\big)\Big)\,\mathrm{d}V^{\otimes N}
\end{equation*}
with
\begin{equation*}
    \mathcal{Z}(\beta) = \int_{M^N} \exp\Big(-\beta\big(E(p_1,\dots,p_N)+\sum_i U(p_i)\big)\Big)\,\mathrm{d}V^{\otimes N},
\end{equation*}
assuming the latter integral converges. We expect our results to hold in this setting, more or less verbatim, as long as the possible failure of convergence of $\mathcal{Z}(\beta)$ arises from the logarithmic short-range nature of the interaction term. That is, as long as the integrability of $\exp\big(-\beta (E+\sum_i U(p_i))\big)\,\mathrm{d}V^{\otimes N}$ is equivalent to local integrability in the interior of $M^N$. One common setting is to take $M$ to be a rectangle, $g$ to be the standard Euclidean metric, $\mathrm{d}V$ to be the standard Lebesgue measure and $U\equiv 0$. 

\subsection{Outline of the Paper}

In \Cref{sec: complex geometry} we recall some notions from complex geometry. In \Cref{sec: partition function} we introduce a complex-geometric formalism and outline some ideas that go into the proofs of \Cref{thm:1,thm:2}. In \Cref{sec: FM} we recall the Fulton--MacPherson compactification of configuration spaces. \Cref{sec: main theorems} contains the proofs of the three main theorems. In \Cref{sec: eigenvalue bounds} we relate \Cref{prob: max,prob: min} to eigenvalue problems and prove \Cref{cor: explicit bounds in the charge case}. In \Cref{sec: random coupling} we use the eigenvalue bounds to obtain bounds in the case of random couplings. In \Cref{sec: examples} we present and prove the results related to the examples discussed in \Cref{sec: applications} above. Finally, in \Cref{sec: other problems} we relate \Cref{prob: max,prob: min} to other discrete optimization problems and point out connections to spin glass models. 

\section{Preliminaries}
\subsection{Complex geometry} 
\label{sec: complex geometry}

We begin this section by recalling some concepts from complex algebraic geometry that are essential to the proofs of \Cref{thm:1,thm:2}. For more details, a good reference is \cite{GH}.

\subsubsection{Divisors}

Let $X$ be a complex manifold of dimension $n$. An ($\mathbb{R}$-)\textit{divisor} $D$ on $X$ is a locally finite formal $\mathbb{R}$-linear combination
\[
    D = \sum_j a_j V_j,
\]
where each $V_j \subset X$ is an irreducible analytic hypersurface. An analytic hypersurface $V$ is locally defined in an open set $U\subset X$ as the zero set of a holomorphic function $g\colon U \to \mathbb{C}$. If $f$ is a holomorphic function on $X$, its associated divisor is 
\[
    \mathrm{Div}(f)\coloneq\sum_V \mathrm{ord}_V(f) V,
\]
where the sum runs over the irreducible components $V$ of the hypersurface defined by $f$, and $\mathrm{ord}_V(f)$ is the order of vanishing of $f$ along $V$. This construction allows divisors to be pulled back via holomorphic maps. Specifically, if $\pi\colon Y \to X$ is a holomorphic map between complex manifolds and $D=\sum_j a_j V_j$ is a divisor on $X$, then
\begin{equation}
    \label{eq: def pullback divisor}
    \pi^*(D)=\sum_j a_j \pi^*(V_j),
\end{equation}
where, locally, $\pi^*(V_j) = \mathrm{Div}(\pi^*(f_j))$ for any local defining function $f_j$ of $V_j$. Finally, for a divisor $D = \sum_j a_j V_j$ on $X$, we define its \textit{support} as the hypersurface
\[
    \mathrm{supp}(D) \coloneq \bigcup_j V_j.
\]

\subsubsection{Blowups}
Again, let $X$ be a complex manifold of dimension $n$, and let $Z \subset X$ be a complex submanifold of codimension $k > 1$. The \textit{blowup of $X$ along $Z$} is a complex $n$-dimensional manifold, denoted by $\mathrm{Bl}_Z X$, together with a holomorphic map $\pi \colon \mathrm{Bl}_Z X \to X$ with the following local description: For any point $p \in Y \subset X$, we can find a neighborhood $B \subset X$ of $p$ and holomorphic coordinates $z = (z_1,\dots,z_n)$ centered around $p$ such that, locally in $B$, $Z=\{z_1 = \dots = z_k = 0\}$. Then, 
\begin{equation*}
    \mathrm{Bl}_{Z} B = \big\{(z,[t])\in B\times\mathbb{CP}^{k-1} \colon z_i t_j = z_j t_i,\, 1\leq i,j\leq k\big\},
\end{equation*}
where $\mathbb{CP}^{k-1}$ denotes complex projective $(k-1)$-space. Geometrically, the blowup replaces $Z$ with a space parameterizing the directions into $Z$ relative to the space $X$, called the \textit{exceptional divisor} $E = \pi^{-1}(Z)$. 
Locally, over $B$, the blowup map $\pi$ is the restriction of the natural projection $B\times \mathbb{CP}^{k-1}\to B$ to the submanifold $\mathrm{Bl}_{Z} B\subset B\times\mathbb{CP}^{k-1}$. By covering $X$ with such coordinate charts, the local constructions glue together to a complex manifold $\mathrm{Bl}_Z X$ and a globally defined holomorphic map $\pi \colon \mathrm{Bl}_Z X\to X$. The map $\pi$ is a biholomorphism over $X\setminus Z$, that is, on the complement of the exceptional divisor.

The space $\mathrm{Bl}_{Z} X$ can be covered by $k$ coordinate charts $\{U_j\}_{j=1}^k$, corresponding to the standard affine charts $\{ t_j \neq 0 \} \subset \mathbb{CP}^{k-1}_{[t]}$, for $j=1,\hdots,k$. If $x= (x_1,\hdots,x_n)$ and $y = (y_1,\hdots,y_n)$ denote holomorphic coordinates in $U_i$ and $U_j$, respectively, then the transition map is given by
\begin{equation*}
    x_\ell = 
    \begin{cases}
        y_i y_j &\text{if } \ell = i, \\
        1/y_j &\text{if } \ell = j, \\
        y_\ell/y_j &\text{if } \ell \neq i,j,\  \ell \leq k, \\
        y_\ell &\text{if } k+1\leq \ell \leq n.
    \end{cases}
\end{equation*}
For each $j=1,\hdots,k$, the restriction of the blowup map to the coordinate chart $U_j$ is given locally by
\begin{equation}
    \label{eq:blowup local picture}
    \pi|_{U_j} \colon (x_1,\hdots,x_n) \mapsto (x_j x_1, \hdots, x_j x_{j-1}, x_j, x_j x_{j+1},\hdots,x_j x_k, x_{k+1},\hdots, x_n).
\end{equation}
From \cref{eq:blowup local picture}, we see that for the pullback of the tuple $(z_1,\hdots,z_k)$, there is locally always one variable that divides the rest. These local functions define the exceptional divisor of the blowup.

\medskip 

Let $V \subset X$ be an analytic subvariety (reduced and irreducible). The \textit{strict transform} of $V$ under the blowup $\pi \colon \mathrm{Bl}_Z X \to X$ is defined by
\[
    \mathrm{Strict}_\pi(V) \coloneq \overline{\pi^{-1}(V\setminus Z)}.
\]
The preimage $\pi^{-1}(V)$ is called the \textit{total transform} of $V$. There is also a third notion, interpolating between these two constructions, called the \textit{dominant transform}. The \emph{dominant transform} of $V$ with respect to $\pi$ is defined as the strict transform whenever $V\not\subset Z$ and the inverse image $\pi^{-1}(V)$ when $V\subset Z$. 

\medskip

Lastly, recall that for a complex manifold $X$ and a Zariski-closed subset $Z \subset X$, an \textit{embedded resolution of singularities} of $Z$ in $X$ is a smooth manifold $\widetilde{X}$ together with a holomorphic map $\pi \colon \widetilde{X} \to X$, which is a composition of blowups along complex submanifolds, such that the restriction $\widetilde{X}\setminus \pi^{-1}(Z) \to X \setminus Z$ is a biholomorphism and such that $\pi^{-1}(Z)$ is a normal crossings hypersurface in $\widetilde{X}$. Recall that a hypersurface is said to have \textit{normal crossings} if, locally, it is a union of coordinate hyperplanes.

\subsubsection{The relative canonical divisor}

Consider a holomorphic map $F \colon Y\to X$ which is a biholomorphism between the complements of closed analytic subsets of $Y$ and $X$ (for instance, a composition of blowups). The \textit{relative canonical divisor} of $F$ is defined by 
\begin{equation}
    \label{eq: def 1 rel can}
    K_{Y/X} \coloneq \mathrm{Div}
    (\mathrm{det}\,\mathrm{Jac}(F)),
\end{equation}
where $\mathrm{Jac}(F)$ denotes the Jacobian matrix of $F$, defined in local coordinates. In this notation, the map $F$ is suppressed and should thus be clear from the context. As a useful alternative, one can also define
\begin{equation}
    \label{eq: def 2 rel can}
    K_{Y/X} = K_Y-\pi^*(K_X). 
\end{equation}
Here, $K_Y=\mathrm{div}(s)$ for a holomorphic section $s$ of the line bundle $\bigwedge^n T^\ast \,Y$, and $K_X$ is defined similarly. Here, $T^\ast \,Y$ denotes the holomorphic cotangent bundle of $Y$. While $K_Y$ and $K_X$ are not uniquely defined as divisors (only as linear equivalence classes), $K_{Y/X}$ is a well-defined divisor once $K_Y$ and $K_X$ are chosen so that they coincide where $F$ is a biholomorphism. This choice will be made implicitly throughout.

As an example, if $\pi \colon \mathrm{Bl}_Z X\to X$ is the blowup of $X$ along $Z$, then 
\begin{equation}
    \label{eq: blowup rel can}
    K_{\mathrm{Bl}_Z X/X}=(\mathrm{codim}(Z)-1)E,
\end{equation}
where $E$ is the exceptional divisor, as is seen from \Cref{eq:blowup local picture}. 

A useful property, which follows from \cref{eq: def 2 rel can}, is that for two maps $F_2 \colon Y_2\to Y_1$ and $F_1 \colon Y_1\to X$ as above, one has
\begin{equation}
    \label{eq: chain rule for rel can div}
    K_{Y_2/X}=K_{Y_2/Y_1}+F_2^*(K_{Y_1/X}),
\end{equation}
for the relative canonical bundle of the composition $F_1\circ F_2 \colon Y_2 \to X$. Note that \cref{eq: chain rule for rel can div} also follows directly from the chain rule for Jacobians applied to \cref{eq: def 1 rel can}. 

\medskip
\subsubsection{Pulling back measures with analytic singularities}
\label{sec: pulling back measures with analytic singularities}

Now, let $X$ be a complex manifold of dimension $n$, and let $D = \sum_j a_j V_j$ be a divisor on $X$. 
Suppose we have a measure $\mu$ on $X$ that locally takes the form
\begin{equation}
    \label{eq: measure with analytic singularities}
    \mu = \prod_j |f_j|^{2a_j} \Psi \,\mathrm{d}z_1\wedge \cdots \wedge\mathrm{d}\bar{z}_n,
\end{equation}
where $f_j$ is a locally defining function of $V_j$, and $\Psi$ is a smooth, nowhere-vanishing function. If $F\colon Y\to X$ is a holomorphic map that is a biholomorphism between the complements of closed analytic subsets $Z \subset Y$ and $W \subset X$, then, locally on $Y\setminus Z$,
\begin{equation}
    \label{eq: local pullback of measure}
    F^*\mu = \prod_j |g_j|^{2b_j} F^*(\Psi) \prod_j |h_j|^{2c_j}\,\mathrm{d}z_1\wedge \cdots\wedge\mathrm{d}\bar{z}_n,
\end{equation}
where $F^*(D)=\sum_j b_j S_j$ and $S_j=\{g_j=0\}$ locally, and $K_{Y/X}=\prod_j c_j K_j$ and $K_j=\{h_j=0\}$ locally. The formula \cref{eq: local pullback of measure} follows from the change of variables formula for integrals, together with the definitions \cref{eq: def pullback divisor} and \cref{eq: def 1 rel can}. Since $\mu$ does not put any mass on $W$, $\mu(U)=F^*\mu(F^*(U\setminus{W}))$ for any measurable $U\subset Y$. By extending $F^*\mu$ by zero over $Z$, we even have $\mu(U)=F^*\mu(F^*(U))$ for any measurable $U$. 
\subsection{Complex geometric formalism}
\label{sec: partition function}
In this section we introduce a complex geometric formalism for the $N$-particle system on the sphere. We begin by identifying $\mathbb{S}^2$ with the Riemann sphere $\mathbb{C}\cup \{\infty\}$ or, equivalently, the complex projective line $\mathbb{CP}^1$. On the standard affine chart $(\mathbb{C})^N\subset(\mathbb{CP}^1)^N$ with coordinates $(z_1,\hdots,z_N)$, the energy in \cref{eq:energy the sphere} can be written as
\begin{equation*}
    E(z_1,\dots,z_N) = \sum_{i<j} c(i,j) \log||z_i-z_j||^{2},
\end{equation*}
where
\[
    ||z_i-z_j||^2\coloneq|z_i-z_j|^2\exp(-\phi(z_i)-\phi(z_j)) =\frac{|z_i - z_j|^2}{(1+|z_i|^2)(1+|z_j|^2)},
\]
and where $|\cdot|$ denotes the standard absolute value on $\C$. The function $\phi(z) = \log(1+|z|^2)$ is the Kähler potential for the Fubini--Study metric on $\mathbb{CP}^1$. Note that
\[
    \mathrm{d}V(z_i)=\frac{1}{{(1+|z_i|^2)}^2}\frac{i}{2\pi}\mathrm{d}z_i\wedge\mathrm{d}\bar{z}_i=\exp(-2\phi(z_i))\frac{i}{2\pi}\mathrm{d}z_i\wedge\mathrm{d}\bar{z}_i.
\]
Letting
\begin{equation}
    \label{eq:Phimetric}
    \Phi(z_1,\dots,z_N)\coloneq-\sum_{i<j} c(i,j) (\phi(z_i)+\phi(z_j)),
\end{equation} 
the Gibbs measure then takes the form
\begin{equation*}
    \mu_\beta = \frac{1}{\mathcal{Z}(\beta)}\prod_{i<j} |z_i-z_j|^{2c(i,j)\beta}\exp(-\beta \Phi)\,\mathrm{d}V^{\otimes N},
\end{equation*}
where the partition function $\mathcal{Z}(\beta)$ is given by
\begin{equation}
    \label{eq:partition_function}
    \mathcal{Z}(\beta) = \int_{\C^N} \prod_{i< j} |z_i -z_j|^{2c(i,j)\beta} \exp({-\beta \Phi})\,\mathrm{d}V^{\otimes N}.
\end{equation}

It suffices to integrate over $\C^N$, since this affine chart is dense in $(\mathbb{CP}^1)^N$. On the remaining standard affine charts covering the reorderings of $\{0\}^k\times\{\infty\}^{N-k}$ for any $k=0,\hdots,N$, the Gibbs measure takes the form
\begin{equation}
    \label{eq:partition_function_other_charts}
    \mu_\beta = \frac{1}{\mathcal{Z}(\beta)}\prod_{1\leq i<j\leq k} |z_i-z_j|^{2c(i,j)\beta} \prod_{\substack{1\leq i \leq k \\ k<j\leq N }} |z_iw_j-1|^{2c(i,j)\beta} \prod_{k<i<j\leq N} |w_i-w_j|^{2c(i,j)\beta} \,\mathrm{d}V^{\otimes N},
\end{equation}
after possibly reordering the variables. Note that the Gibbs measure is of the type \cref{eq: measure with analytic singularities} in \Cref{sec: pulling back measures with analytic singularities}.

The partition function is convergent only on a (possibly empty) interval or ray $(\beta^-,\beta^+) \subset\mathbb{R}$, depending on the couplings $c(i,j)$ and on $N$, and divergent otherwise. For any $\beta \in \mathbb{R}$, the local integrands in \cref{eq:partition_function} and \cref{eq:partition_function_other_charts} are locally integrable outside
\[
    D \coloneq \bigcup_{i<j} \{p_i=p_j\}.
\]
Consequently, for any $\beta \in \mathbb{R}$ and any smooth function $\xi$ with compact support in $(\mathbb{CP}^1)^N \setminus D$,
\begin{equation}
    \label{eq:Z_distribution}
    \mathcal{Z}^\xi(\beta) = \int_{(\mathbb{S}^2)^N} \prod_{i<j} d(p_i,p_j)^{2c(i,j) \beta} \xi \,\d V^{\otimes N},
\end{equation}
is defined, and the map $\xi\mapsto \mathcal{Z}^\xi(\beta)$ defines a distribution on $(\mathbb{CP}^1)^N\setminus D$. 

\medskip

When $c(i,j) \in \mathbb{Q}_+$ for each $i\neq j$, \cref{eq:Z_distribution} is an example of an Archimedean local zeta function, see, e.g., \cite{I1}. Note that, for $\beta \in (\beta^-,\beta^+)$, \cref{eq:Z_distribution} defines a distribution on all of $(\mathbb{CP}^1)^N$. This perspective of viewing the partition function as a distribution in the above sense is useful in relation to \Cref{thm:2}, which concerns the weak limit of the Gibbs measure $\mu_\beta$ as we approach the critical inverse temperatures. Note that, for $\beta \in (\beta^-,\beta^+)$,
\begin{equation}
    \label{eq:actionofgibbs}
    \langle \mu_\beta,\xi\rangle \coloneq\int_{(\mathbb{S}^2)^N} \xi \,\mu_\beta = \frac{\mathcal{Z}^\xi(\beta)}{\mathcal{Z}^1(\beta)}.
\end{equation}

If we allow for complex values of the inverse temperature $\beta$, then $\mathcal{Z}^\xi(\beta)$ defines a distribution on $(\mathbb{CP}^1)^N$ for any $\beta$ in the vertical strip $\Omega = \{\beta\in\C : \beta^- < \mathfrak{Re}\,\beta < \beta^+\}$. Moreover, $\mathcal{Z}^\xi(\beta)$ is holomorphic as a function of $\beta$ in $\Omega$. By a classical result, originally due to Bernstein and Gelfand in \cite{BG} and independently Atiyah in \cite{Atiyah}, $\mathcal{Z}^\xi(\beta)$ has a meromorphic continuation to all of $\C_\beta$ with poles in a discrete subset of $\mathbb{Q}$, for any test function $\xi$. This classical result applies when $c(i,j) \in \mathbb{Q}_+$, but as we show below, a version of it extends to our more general setting of mixed signs. 

In particular, $\beta^-$ and $\beta^+$ will be among the possible poles of $\mathcal{Z}^\xi(\beta)$. Thus, we can consider the Laurent series expansion of $\mathcal{Z}^\xi(\beta)$ about $\beta^{-}$ (and, similarly, about $\beta^+$),
\[
    \mathcal{Z}^\xi(\beta) = \sum\limits_{\ell \geq -\kappa} c_\ell (\beta-\beta^{-})^\ell,
\]
where $0\leq \kappa \leq N$. 
The leading order coefficient $c_{-\kappa}$, viewed as a function of $\xi$, defines a positive distribution, that is, a measure, on ${(\mathbb{CP}^1)^N}$. 
This measure, up to normalization, is equal to the weak limit of the Gibbs measure $\mu_\beta$ as $\beta \to \beta^-$, see the proof of \Cref{thm:2}.

\medskip 

Let us consider the simplest possible example, when there are only two particles.
\begin{example}
    \label{ex:N=2}
    Let $N=2$, $c\coloneq c(1,2)>0$ and let $\xi$ be a test function on $\mathbb{C}^2\subset (\mathbb{CP}^1)^2$. Then $\mathcal{Z}^\xi(\beta)$ becomes
    \[
        \mathcal{Z}^\xi(\beta) = \int_{\C^2} |z_1 - z_2|^{2\beta c} \xi e^{-\beta U} \frac{i^2}{(2\pi)^2} \frac{\d z_1 \wedge \d \bar{z}_1 \wedge \d z_2 \wedge \d \bar{z}_2}{{(1+|z_1|^2)}^2 {(1+|z_2|^2)}^2},
    \]
    It is clear that $\beta^+ = \infty$, since $c>0$, and that $\beta^- = -1/c$. Moreover, as $\beta$ approaches $-1/c$, a standard computation in residue theory, see the proof of \Cref{thm:2} below, shows that
    \begin{align*}
        \mathcal{Z}^\xi(\beta) = \frac{1}{(\beta+1/c)} \int_{\{ z_1 - z_2 = 0 \}} \xi e^{U}\,\d \widetilde{V} + \mathcal{O}(1),
    \end{align*}
    where
    \[
        \d \widetilde{V} = \frac{i}{2\pi}\frac{\d \tfrac{1}{2}(z_1 + z_2) \wedge \d \tfrac{1}{2}(\bar{z}_1 + \bar{z}_2))}{{\big(1+ |\tfrac{1}{2}(z_1 + z_2)|^2 \big)}^4}.
    \]
    Thus, the coefficient $c_{-1}$ of the Laurent series expansion of $\mathcal{Z}^\xi(\beta)$ about $-1/c$ is given by the action of $\delta_0(z_1 - z_2) \wedge  e^U \d \widetilde{V}$ on $\xi$, where $\delta_0$ is the Dirac measure on $\C$.
\end{example}
When $N\geq 3$, that is, for three or more interacting particles, the situation is more involved than in \Cref{ex:N=2}. The main difficulty arises from the fact that the singular locus of the integrand is not a normal crossings hypersurface. Consider an integral of the form
\begin{equation}
    \label{eq:simpleform}
    I(\beta) = i^N \int_{\C^N} |z_1|^{2(\beta a_1 + b_1)} \cdots |z_\kappa|^{2(\beta a_\kappa + b_\kappa)} \xi \,\d z_1 \wedge \d \bar{z}_1 \wedge \cdots \d z_N \wedge \d \bar{z}_N,
\end{equation}
whose singular set is the normal crossings hypersurface $\{z_1 \cdots z_\kappa = 0\}$. In this case, one can repeat the arguments of \Cref{ex:N=2} in each variable separately (cf.\ the proof of \Cref{thm:2}) to find the weak limit of the Gibbs measure. For $N\geq 3$, however, $\mathcal{Z}^\xi(\beta)$ is not locally of the form \cref{eq:simpleform}. Nevertheless, it turns out that we can reduce to this situation by pulling back the integral defining $\mathcal{Z}^\xi(\beta)$ via a suitable holomorphic
map. We showcase this for $N=3$ with the following example.
\begin{example}
    \label{ex:N=3blowup}
    Let us fix $N = 3$ and assume for simplicity that $c(i,j)>0$ for each $i\neq j$. Just as in \Cref{ex:N=2}, we then have that $\beta^+=\infty$. In the standard affine subspace $\mathbb{C}^3 \subset (\mathbb{CP}^1)^3$, we can write $\d V^{\otimes 3} = i^3 f \,\d z \wedge \d \bar{z}$, where $f = (\tfrac{1}{2\pi})^3 \exp(-2\sum_{j=1}^3 \phi(z_j))$, $\phi(z_j) = \log(1+|z_j|^2)$ and where $\d z \wedge \d \bar{z} = \d z_1 \wedge \d \bar{z}_1 \wedge \d z_2 \wedge \d \bar{z}_2 \wedge \d z_3 \wedge \d \bar{z}_3$. Consider the change of variables
    \begin{equation*}
        w_1 = z_1 - z_2, \quad w_2 = z_2 - z_3,\quad w_3 = \frac{1}{3}(z_1 + z_2 + z_3),
    \end{equation*}
    which satisfies $\d w \wedge \d \bar{w} = \d z \wedge \d \bar{z}$. For $\mathfrak{Re}\,\beta > 0$, and for any test function $\xi$, we have that
    \[
        \mathcal{Z}^\xi_3(\beta) = i^3\int_{\C^3} |w_1|^{2 c(1,2)\beta} |w_2|^{2 c(2,3)\beta} |w_1 + w_2|^{2 c(1,3)\beta} e^{-\beta U} \xi f \,\d w \wedge \d \bar{w}.
    \]
    To determine $\beta^-$, we first address the fact that
    \[
        \{w_1 = 0\} \cup \{w_2 = 0\} \cup \{w_1 + w_2 = 0\}
    \]
    is not a normal crossings hypersurface. To remedy this, we consider the blowup $\pi \colon \mathrm{Bl}_W \C^3 \to \C^3$ along the locus $W = \{w_1 = w_2 = 0\}$. The space $\mathrm{Bl}_W \C^3$ can be covered by two coordinate charts, and the blowup map $\pi$ can be described explicitly on each. More precisely, there exists an open covering $\{U_1,U_2\}$ of $\mathrm{Bl}_W \C^3$ such that, in local holomorphic coordinates
    \[
        (x_1,x_2,x_3) \quad \text{on } U_1 \simeq \C^3, \qquad (y_1,y_2,y_3) \quad \text{on } U_2 \simeq \C^3,
    \]
    the transition relations on $U_1\cap U_2$ are given by
    \[
        x_1 = y_1 y_2,\quad x_2 = 1/y_1,\quad x_3 = y_3.
    \]
    The blowup map can be written locally as
    \[
        \pi(x_1,x_2,x_3) = (x_1, x_1 x_2, x_3)\quad \text{in } U_1, \qquad \pi(y_1,y_2,y_3) = (y_2 y_1, y_2, y_3)\quad \text{in } U_2,
    \]
    cf.\ \cref{eq:blowup local picture}. In the chart $U_1$, we have
    \[
            \begin{aligned}
                &\pi^*\Big(  |w_1|^{2 c(1,2)\beta} |w_2|^{2 c(2,3)\beta} |w_1 + w_2|^{2 c(1,3)\beta} e^{-\beta U} \xi f \,\d w\wedge \d \bar{w}\Big) = \\
                &\qquad\qquad\qquad = |x_1|^{2(a_1\beta + b_1)}|x_2|^{2(a_2\beta + b_2)} |1+x_2|^{2(a_3\beta+b_3)} \pi^*e^{-\beta U}\pi^*\xi \pi^* f \,\d x\wedge \d \bar{x},
            \end{aligned}
    \]
    where
    \[
        a_1 = c(1,2)+c(2,3)+c(1,3),\quad a_2 = c(2,3), \quad a_3 = c(1,3),
    \]
    and
    \[
        b_1 = 1, \quad b_2 = b_3 = 0.
    \]
    The integrand is locally integrable in $U_1$ provided that
    \begin{align*}
        \beta &> \max\Big\{ -\frac{1+b_1}{a_1}, -\frac{1+b_2}{a_2},-\frac{1+b_3}{a_3} \Big\} \\
        &= \max\Big\{ -\frac{2}{c(1,2)+c(2,3)+c(1,3)}, -\frac{1}{c(2,3)},-\frac{1}{c(1,3)} \Big\}.
    \end{align*}
    By symmetry, an analogous integrability condition holds in $U_2$. Since $\pi$ is proper, we conclude that
    \[
        \beta^- = \max\Big\{ -\frac{2}{c(1,2)+c(2,3)+c(1,3)},- \frac{1}{c(1,2)}, -\frac{1}{c(2,3)},-\frac{1}{c(1,3)} \Big\}.
    \]
    With more work, following \Cref{ex:N=2}, one can also deduce the leading asymptotics of the Gibbs measure and the partition function using this blowup construction.
\end{example}

Reducing to the case \cref{eq:simpleform} can be done locally, for general $N$, by means of an embedded resolution of singularities of the pair $((\mathbb{CP}^1)^N,D)$. In our setting, there is an explicit construction of such a resolution for arbitrary $N$, due to Fulton and MacPherson \cite{FM}; see the following section for details. 
\begin{remark}
    The techniques of this paper extend naturally to other Gibbs measures, for instance to variants of Log gases on higher-dimensional varieties $X$. Central to the proofs of \Cref{thm:1,thm:2} is the meromorphic continuation of the partition function as a distribution-valued function of $\beta$, a general version of which can be found in \cite{Sv}. The applicability of the Fulton--MacPherson compactification, which provides an explicit embedded resolution of singularities, relies on the fact that the singularities of the integrand of the partition function are supported on the big diagonal $\bigcup_{i<j}\{x_i=x_j\} \subset X^N$, which is typically the case for pairwise interactions. A related result for Koba--Nielsen string amplitudes over an arbitrary local field of characteristic zero, which closely resemble partition functions of Log gases, can be found in \cite{Bo-GaVeZu-Ga}.
\end{remark}
\subsection{The Fulton--MacPherson compactification of configuration space}
\label{sec: FM}

In \cite{FM}, Fulton and MacPherson introduce a remarkable compactification of the configuration space of a nonsingular algebraic variety. Given a nonsingular algebraic variety $X$ and a natural number $N$, they construct a nonsingular algebraic variety $X^{[N]}$ along with a proper map $\pi\colon X^{[N]}\to X^N$. The map $\pi$ restricts to an isomorphism over the open subset $X^N\setminus \bigcup_{i<j}W_{\{i,j\}}$ where $W_{\{i,j\}}$ is the diagonal $\{(x_1,..,x_N)\in X^N: x_i=x_j\}$, that is, $\pi$ is an isomorphism outside of the degenerate configurations where at least two particles coincide. Moreover, $\pi^{-1}(\bigcup_{i<j}W_{\{i,j\}})$ is a normal crossings divisor in $X^{[N]}$. For simplicity, we will assume that $X$ has complex dimension $1$, since our application concerns $X=\mathbb{CP}^1$.

\medskip

The construction of $X^{[N]}$ is somewhat involved, but it can be constructed as a sequence of blowups $\pi_W$ ,where $W$ ranges over the set of subvarieties
\begin{equation}
    \label{eq:set of diagonals}
    \mathcal{G}\coloneq \big\{ \{ x \in X^N \colon x_{i_1}=\cdots= x_{i_k} \} \colon 1\leq i_1<\cdots < i_k \leq N, k=2,\dots,N \big\}
\end{equation}
of $X^N$, see below. It is convenient to identify $\mathcal{G}$ with the set $G$ of subsets of $\{1,\dots,N\}$ of size at least two, via $\{i_1,\dots,i_k\}\leftrightarrow \{x_{i_1}=\cdots=x_{i_k}\}$. We write $W_S \in \mathcal{G}$ for the subvariety corresponding to the subset $S \in G$ under this bijection. We also denote by $\mathcal G_j$ the subset of $\mathcal{G}$ consisting of elements
\[
    W_{\{i_1,\dots,i_r\}} = \{ x \in X^N \colon x_{i_1} = \cdots = x_{i_r} \}
\]
with $i_1,\dots,i_r\leq j$. 

\medskip

For the reader's convenience, we briefly recall the construction of $X^{[N]}$, in the language of \cite{Li}. The variety $X^{[N]}$ is obtained inductively through a sequence of blowups. Starting from a variety $Y_n$ isomorphic to $X^{[n]}\times X^{N-n}$, one constructs a new variety $Y_{n+1}$ isomorphic to $X^{[n+1]}\times X^{N-(n+1)}$ is constructed. This process begins with $Y_1\coloneq X^{[1]} \times X^{N-1} = X^N$ and terminates with $Y_N\coloneq X^{[N]}$. Moreover, this construction yields a factorization
\begin{equation}
    \pi=\pi_N\circ\cdots\circ \pi_2 \colon Y_N\to \cdots \to Y_1 = (\mathbb{CP}^1)^N.
\end{equation}
For each $n\geq 1$, the space $Y_{n+1}$ is obtained from $Y_n$ by a sequence of blowups. More precisely, we define
\begin{equation}
    \pi_{n+1}=\pi_{n+1,n}\circ\cdots\circ \pi_{n+1,1} \colon Y_{n+1}=Y_{n+1,n}\to \cdots \to Y_{n+1,0}=Y_{n},
\end{equation}
where, for each $1\leq k \leq n$, the map $\pi_{n+1,k}$ is the iterated blowup along the iterated dominant transforms (with respect to all preceding blowups $\pi_{2,1},\pi_{3,1},\pi_{3,2},\hdots$, $\pi_{n+1,1},\hdots,\pi_{n+1,k-1}$) of all of $W\in \mathcal{G}_{n+1}$ with $\mathrm{codim}(W)=n+1-k$. These subvarieties are all pairwise disjoint, so the order of the blowups at this stage is irrelevant. For a given $W\in \mathcal{G}$, we denote by $W^{n+1,k}$ the iterated dominant transform of $W$ under the maps $\pi_{m,\ell}$ with $1\leq m \leq n$ and $1 \leq \ell < k$, so that $W^{n+1,k} \subset Y_{n+1,k}$.

\medskip

Each iterated blowup $\pi_{j,k}$ gives rise to collection of exceptional divisors, which we denote by $E_W^{j,k}$. For any $m\geq j$ and $\ell\geq k$, we denote by $E_W^{m,\ell}$ the strict transform of these divisors with respect to all subsequent blowups up to level $m,\ell$. We include in this notation also the blowups along subvarieties of codimension $1$; although these are isomorphisms, they nevertheless give rise to exceptional divisors, namely themselves, which simplifies the notation. 

On $X^{[n]}\times X^{N-n}$ we write $E_W^n\coloneq E_W^{n,n-1}$, and finally on $X^{[N]}$ we denote $E_W\coloneq E_W^{N,N-1}$, which is defined for every $W\in \mathcal{G}$. In fact, the preimage of $\bigcup_{i<j}W_{\{i,j\}}$ under $\pi$ is precisely the union of these divisors:
\[
    \pi^{-1}\bigg(\bigcup_{i<j} W_{\{i,j\}}\bigg) = \bigcup_{W\in\mathcal{G}} E_W.
\]
By \cite{FM}, $\pi^{-1}(\bigcup_{i<j} W_{\{i,j\}})$ has normal crossings. Hence, $\pi$ is an embedded resolution of singularities of the pair $((\mathbb{CP}^1)^N, \bigcup_{i<j} W_{\{i,j\}})$. 

\medskip 

Now let $X=\mathbb{CP}^1$. To keep track of the singularities of the integrand in \cref{eq: gibbs measure}, we introduce an $\mathbb{R}$-divisor on $(\mathbb{CP}^1)^N$,
\begin{equation}
    \label{eq:Dc}
    D_c \coloneq \sum_{i <j} c(i,j)W_{\{i,j\}}.
\end{equation}
The following two lemmas will be useful for understanding how the integrand behaves under the pullback to the Fulton--MacPherson compactification. 
\begin{lemma}
    \label{lemma: pullback div}
    Let $\pi\colon(\mathbb{CP}^1)^{[N]}\to (\mathbb{CP}^1)^N$ be the Fulton--MacPherson compactification. Then
    \begin{equation}
        \label{eq: order of Dc W}
        \pi^*(D_c) = \sum_{W\in \mathcal{G}} \sum_{\substack{i<j:\\W_{\{i,j\}} \supseteq W}} c(i,j)E_W,
    \end{equation}
    where $E_W$ is the exceptional divisor corresponding to the diagonal $W$. 
\end{lemma}
\begin{proof}
    By \cite[Theorem 3(3)]{FM}, the scheme-theoretic inverse image of a diagonal $W$ is the union of the exceptional divisors $E_V$ with $V\subset W$. Hence,
    \begin{equation}
        \pi^*(D_c) = \sum_{i<j} c(i,j)\pi^*(W_{\{i,j\}}) = \sum_{i<j} \sum_{W\subseteq W_{\{i,j\}}} c(i,j) E_W = \sum_{W\in \mathcal{G}} \sum_{\substack{i<j:\\W_{\{i,j\}}\supseteq W}} c(i,j) E_W,
    \end{equation}
    as claimed.
\end{proof}
In view of the bijection between $\mathcal{G}$ and $G$, the result above can be rewritten as
\begin{equation}
        \label{eq: order of Dc S}
        \pi^*(D_c) = \sum_{S\in G} \sum_{i<j\in S} c(i,j)E_S,
\end{equation}
where $E_S$ denotes the exceptional divisor $E_{W_S}$ corresponding to $S$.
\begin{lemma}
    \label{lemma: rel can div}
    Let $\pi\colon(\mathbb{CP}^1)^{[N]}\to (\mathbb{CP}^1)^N$ be the Fulton--MacPherson compactification. The relative canonical divisor
    \begin{equation}
        \label{eq: rel can div W}
        K_{(\mathbb{CP}^1)^{[N]} / (\mathbb{CP}^1)^N} = \sum_{W\in\mathcal{G}} (\mathrm{codim}(W)-1)E_W,
    \end{equation}
    where $E_W$ is the exceptional divisor corresponding to $W\in\mathcal{G}$.
\end{lemma}
\begin{proof}
    The proof proceeds by induction on $n$, using repeatedly the formula for the relative canonical bundle of a blowup \cref{eq: blowup rel can} together with the chain rule \cref{eq: chain rule for rel can div}. As a base case, take $N=2$. Since $(\mathbb{CP}^1)^{[2]}=(\mathbb{CP}^1)^2$,  $K_{(\mathbb{CP}^1)^{[2]} / (\mathbb{CP}^1)^2}=0$. Moving on to the induction step, first, we introduce some notation. Recall that
    \[
        \mathcal{G}_n = \{ W_S \colon  S\cap \{n+1,\hdots,N\} = \emptyset, \,|S|\geq 2 \}.
    \]
    Thus,
    \begin{align*}
        \mathcal{G}_{n+1} \setminus \mathcal{G}_n = \{ W_S \colon S\cap \{n+1,\hdots,N\} = \{n+1\},\, |S| \geq 2\}.
    \end{align*}
    Recall the iterated blowup construction of $(\mathbb{CP}^{1})^{[N]}$ described above, and define
    \[
        \pi_{[n]} = \pi_n \circ \cdots \circ \pi_2 \colon Y_n\to \cdots \to Y_2 \to Y_1 = (\mathbb{CP}^1)^{N},
    \]
    where
    \[
        \pi_j = \pi_{j,j-1}\circ \cdots \circ \pi_{j,1}\colon Y_j = Y_{j,j-1} \to \cdots \to Y_{j,1} \to Y_{j,0} = Y_{j-1},
    \]
    and where $\pi_{j,k}\colon Y_{j,k} \to Y_{j,k-1}$ is the iterated blowup along the iterated dominant transforms (with respect to the previous blowups $\pi_{2,1},\pi_{3,1},\pi_{3,2},\hdots,\pi_{j,1},\hdots,\pi_{j,k-1}$) of all $W \in \mathcal{G}_j$ with $\mathrm{codim}\, W = j-k$. Thus, the Fulton--MacPherson compactification is obtained as $\pi \coloneq \pi_{[N]} \colon  (\mathbb{CP}^1)^{[N]}\to (\mathbb{CP}^1)^N$.

    \medskip

    For the induction hypothesis, suppose that
    \[
        K_{Y_n}-\pi_{[n]}^*(K_{(\mathbb{CP}^1)^{N}}) = \sum_{W\in\mathcal{G}_n} (\mathrm{codim}(W)-1)E_W^n.
    \]
    We then compute
    \begin{equation}
        \label{eq: split of rel canonical comp}
        \begin{aligned}
            K_{Y_{n+1}} - \pi_{[n+1]}^*(K_{(\mathbb{CP}^1)^{N}}) & = K_{Y_{n+1}} - (\pi_{n+1} \circ \pi_{[n]})^*(K_{(\mathbb{CP}^1)^{N}}) \\
            &= K_{Y_{n+1}} - \pi_{n+1}^*\bigg(K_{Y_n} -\sum_{W\in\mathcal{G}_n} (\mathrm{codim}(W)-1)E_W^n\bigg) \\
            &= K_{Y_{n+1}} - \pi_{n+1}^*(K_{Y_n}) \\
            &\qquad\qquad + \sum_{W\in\mathcal{G}_n} (\mathrm{codim}(W)-1)\pi_{n+1}^*(E_W^n).
        \end{aligned}
    \end{equation}
    To proceed, we recall the following fact from \cite[Proposition 3.4]{FM}:
    \begin{equation}
        \pi_{n,k}^*(E^{n,k-1}_W)=E^{n,k}_W
    \end{equation}
    unless $W=W_S$ with $n\notin S$ and $|S|=n-k$ in which case
    \begin{equation}
        \pi_{n,k}^*(E^{n,k-1}_W) = E^{n,k}_W+E^{n,k}_{W'},
    \end{equation}
    where $W'=W_{S'}$ with $S'=S\cup\{n\}$. 

    Using this, we compute
    \begin{align*}
        K_{Y_{n+1}}-\pi_{n+1}^*(K_{Y_n})&= K_{Y_{n+1}}-(\pi_{n+1,n+1}\circ\cdots\circ\pi_{n+1,1})^* (K_{Y_n})\\
        &=\sum_{W\in\mathcal{G}_{n+1}\setminus \mathcal{G}_n} (\mathrm{codim}(W^{n,k})-1)E_{W}^{n+1} \\
        &=\sum_{\substack{W\in\mathcal{G}_{n+1}\setminus \mathcal{G}_n:\\\mathrm{codim}(W)>1}} E_{W}^{n+1}
    \end{align*}
    where, in the last step, we use that $W^{n,k}$ has codimension $2$ whenever $\mathrm{codim}(W)>1$ and codimension $1$ otherwise, see \cite[Proposition 3.1]{FM}.
    Next we compute
    \begin{align*}
        \sum_{W\in\mathcal{G}_n} (\mathrm{codim}(W)-1)\pi_{n+1}^*(E_W^n) &= \sum_{W\in\mathcal{G}_n} (\mathrm{codim}(W)-1)(E^{n+1}_W+E^{n+1}_{W'})\\
        &= \sum_{W\in\mathcal{G}_n} (\mathrm{codim}(W)-1)E^{n+1}_W\\
        &\qquad\qquad + \sum_{W\in\mathcal{G}_{n+1}\setminus\mathcal{G}_n} (\mathrm{codim}(W)-2)E^{n+1}_W,
    \end{align*}
    where $W'$ is given by $W_{S'}$ where $S'=S\cup\{n\}$ and $W=W_S$. 
    
    Going back to \eqref{eq: split of rel canonical comp} and putting it all together we get, 
    \begin{align*}
        K_{Y_{n+1}} - \pi_{[n+1]}^*(K_{(\mathbb{CP}^1)^{N}}) &= K_{Y_{n+1}} - \bigg(K_{Y_n+1}-\sum_{\substack{W\in\mathcal{G}_{n+1}\setminus \mathcal{G}_n:\\\mathrm{codim}(W)>1}} E_{W}^{n+1}\bigg) \\
        &\qquad\qquad+\sum_{W\in\mathcal{G}_n} (\mathrm{codim}(W)-1)E^{n+1}_W \\
        &\qquad\qquad + \sum_{W\in\mathcal{G}_{n+1}\setminus\mathcal{G}_n} (\mathrm{codim}(W)-2)E^{n+1}_W \\
        &= \sum_{W\in \mathcal{G}_{n+1}} (\mathrm{codim}(W)-1) E_W^{n+1}.
    \end{align*}
\end{proof}

In view of the bijection between $\mathcal{G}$ and $G$, the above result can be written as
\begin{equation}
        \label{eq: rel can div S}
        K_{(\mathbb{CP}^1)^{[N]}} - \pi^*(K_{(\mathbb{CP}^1)^N}) = \sum_{S\in G} (|S|-2)E_S.
\end{equation}
For \Cref{thm:2,thm:3} we will also need to know which exceptional divisors $E_W$ on $(\mathbb{CP}^1)^{[N]}$ intersect each other. To this end we have the following description from \cite{FM}.
\begin{lemma}[{\cite[Theorem 3(2)]{FM}}]
    \label{lemma: intersection pattern}
    The intersection 
    \begin{equation*}
        E_{S_1}\cap\cdots \cap E_{S_k},
    \end{equation*}
    for subsets $S_1,\dots,S_k\in G$, is non-empty if and only if for each pair $1\leq i,j\leq k$ either one of $S_i$ and $S_j$ is contained in the other, or $S_i \cap S_j = \emptyset$. 
\end{lemma}
\section{Proofs of the main theorems}
\label{sec: main theorems}

\begin{proof}[Proof of \Cref{thm:1}]
    Let $\pi\colon(\mathbb{CP}^1)^{[N]}\to (\mathbb{CP}^1)^N$ be as in \Cref{sec: FM}, that is, it arises from the Fulton-MacPherson compactification of the configuration space of $\mathbb{CP}^1$. Since $\pi$ restricts to a biholomorphism over $(\mathbb{CP}^1)^N\setminus D$, and since $D=\bigcup_{i<j}W_{\{i,j\}}$ has codimension 1,
    the finiteness of the partition function \cref{eq: partition function intro} is reduced to the integrability of the measure
    \begin{equation*}
        (\pi|_{(\mathbb{CP}^1)^{[N]}\setminus \pi^{-1}(D)})^*\bigg(\exp(-\beta (E(p_1,\dots,p_N))\,\d V^{\otimes N}\bigg),
    \end{equation*}
    extended by zero on $(\mathbb{CP}^1)^{[N]}$. The map $\pi$ being an embedded resolution of singularities of $\bigcup_{i<j}W_{\{i,j\}} \subset (\mathbb{CP}^1)^N$ implies that we can cover $(\mathbb{CP}^1)^{[N]}$ with coordinate charts on which the hypersurface $\mathrm{supp}(\pi^*(D))\cup\mathrm{supp}(K_{(\mathbb{CP}^1)^{[N]}/(\mathbb{CP}^1)^N})=\mathrm{supp}(\pi^*(D))$ has normal crossings. Thus, for any such chart $U$, we can find holomorphic coordinates $w = (w_1,\hdots,w_n)$ such that
    \begin{equation}
        \label{eq:pullbackofgibbs}
        (\pi|_{(\mathbb{CP}^1)^{[N]}\setminus \pi^{-1}(D)})^*\bigg(\exp(-\beta E)\,\d V^{\otimes N}\bigg)\bigg|_U = \prod_{\ell=1}^k |w_\ell|^{2a_\ell\beta+2b_\ell} e^{-\beta \pi^* \Phi}\,\mathrm{d}\widetilde{V}^{\otimes N},
    \end{equation}
    locally in $U$, where $0\leq k\leq N$, $a_\ell$ and $b_\ell$ are real numbers, $\mathrm{d}\widetilde{V}^{\otimes N}$ is a volume form on $U$ and $\Phi$ is (locally) defined by \cref{eq:Phimetric}. For any $\ell=1,..,k$, the hypersurface $\{w_\ell=0\}$ on $U$ corresponds to the restriction to $U$ of $E_W$ for some $W\in \mathcal{G}$, or, equivalently, to some $E_S$ for $S\in G$. By definition, $a_\ell$ is the coefficient in front of $E_W$ in \cref{eq: order of Dc W} in \Cref{lemma: pullback div}. Similarly, $b_\ell$ is the coefficient in front of $E_W$ in \cref{eq: rel can div W} in \Cref{lemma: rel can div}, cf.\ \Cref{sec: pulling back measures with analytic singularities}. Thus, $a_\ell=\sum_{i<j: W_{\{i,j\}} \supseteq W} c(i,j)$ and $b_\ell=\mathrm{codim}(W)-1$ by \Cref{lemma: pullback div,lemma: rel can div}. The finiteness of the partition function reduces to the inequalities
    \begin{equation}
        2\beta \sum_{\substack{i<j:\\W_{\{i,j\}} \supseteq W}} c(i,j)+2(\mathrm{codim}(W)-1)>-2,
    \end{equation}
    that should hold for all $W\in\mathcal{G}$. In terms of the bijection to subsets of $\{1,\dots,N\}$, the finiteness is equivalent to the inequalities
    \begin{equation*}
        2\beta\sum_{i<j\in S}c(i,j)+2(|S|-2)>-2,
    \end{equation*}
    that should hold for all $S\subset \{1,\dots,N\}$ of size at least $2$. These inequalities correspond to lower or upper bounds for $\beta$ depending on whether $\sum_{i<j\in S}c(i,j)$ is positive or negative. The sharp bounds that characterize finiteness are then given in terms of optimization problems over subsets of size at least $2$ of $\{1,\dots,N\}$, corresponding precisely to Problem \ref{prob: max} and Problem $\ref{prob: min}$, yielding the statement.
\end{proof}
The idea behind the proof of \Cref{thm:2} is to show that, for any test function $\xi$, $\mathcal{Z}^\xi(\beta)$ defined in \cref{eq:Z_distribution} admits a meromorphic continuation to neighborhoods of the critical values $\beta^+$ and $\beta^-$. We then analyze the Laurent series expansions of $\mathcal{Z}^\xi(\beta)$ about these points. When all coupling coefficients $c(i,j)$ are positive integers, the existence of such an expansion, particularly in a local setting, is a classical result of Atiyah \cite{Atiyah} and Bernstein--Gelfand \cite{BG}, with several known generalizations. In our setting, we provide a direct proof that follows closely the classical arguments. 

\medskip

\begin{proof}[Proof of \Cref{thm:2}]
    We divide proof into three steps. In Step 1, we show that $\mathcal{Z}^\xi(\beta)$ admits a meromorphic continuation to all of $\mathbb{C}$, with a discrete set of poles along the real line. This is a slight variation of a classical result, and the techniques involved in this part of the proof are standard. In Step 2, we show that the lowest-order nonvanishing coefficient in the Laurent expansion of $\mathcal{Z}^\xi(\beta)$ about $\beta^\pm$ defines the action of a positive measure on the test function $\xi$. Furthermore, we show that the weak convergence of the Gibbs measure as $\beta \to \beta^\pm$ follows from this fact. Lastly, in Step 3, we determine the support of the limiting measure. We prove \Cref{thm:2} for $\beta \to \beta^-$; the case $\beta \to \beta^+$ is completely analogous and, in fact, follows from the proof for $\beta\to \beta^-$ by replacing $\beta$ with $-\beta$.

    \medskip 
    \noindent
    \textbf{Step 1: $\mathcal{Z}^\xi(\beta)$ has a meromorphic continuation to all of $\mathbb{C}$}.

    Let $\pi \colon (\mathbb{CP}^1)^{[N]} \to (\mathbb{CP}^1)^N$ be the Fulton--MacPherson compactification. In particular, the divisor $\mathrm{supp}(\pi^*(D_c))\cup\mathrm{supp}(K_{(\mathbb{CP}^1)^{[N]}/(\mathbb{CP}^1)^N})$ has normal crossings. Recall, from the proof of \Cref{thm:1}, that we can find an open cover of $(\mathbb{CP}^1)^{[N]}$, where, in each chart, the pullback of the Gibbs measure takes the form \cref{eq:pullbackofgibbs}. Let $\{U_j\}$ be such a cover of $\mathrm{supp}\,\pi^*\xi \subseteq (\mathbb{CP}^1)^{[N]}$, and let $\{\rho_j\}$ be a partition of unity subordinate to this cover. For $\beta^- < \mathfrak{Re}\,\beta < \beta^+$, it then follows that $\mathcal{Z}^\xi(\beta)$ is a finite (since $\pi$ is proper) sum of integrals of the form
    \begin{equation}
        \label{eq:Zlocally}
        I_j(\beta) = \int_{U_j}|w_1|^{2(a^j_1 \beta + b^j_1)} \cdots |w_{k^j}|^{2 (a^j_{k^j} \beta + b^j_{k^j})} \Psi^j(\beta) \,\d \widetilde{V}^{\otimes N},
    \end{equation}
    where $a^j_i,b^j_i \in \mathbb{R}$, $0\leq k^j \leq N$, where
    \begin{equation}
        \label{eq:Psibeta}
        \Psi^j(\beta) = \rho_j \pi^* \xi \exp\big({-\beta \pi^* \Phi} \big)
    \end{equation}
    is smooth and compactly supported in $U_j$, uniformly in $\beta$, and where $\d \widetilde{V}^{\otimes N}$ is a volume form on $(\mathbb{CP}^1)^{[N]}$. Suppressing the dependence on the chart $U_j$ in the notation in the sequel, if $k = 0$, then $I(\beta)$ is defined and holomorphic for all $\beta \in \C$, so we can assume that $k\geq 1$. Moreover, without loss of generality, we may assume that $a_1,\hdots,a_\ell > 0$ and $a_{\ell+1},\hdots,a_k < 0$, for some $1\leq \ell < k$, and then, in view of \cref{eq:Zlocally}, conclude that $I(\beta)$ is defined and holomorphic for
    \begin{equation}
        \label{eq:localholomorphicitybounds}
        \max_{1\leq j \leq \ell} -\frac{1+b_j}{a_j} < \mathfrak{Re}\,\beta < \min_{\ell+1\leq j \leq k} -\frac{1+b_j}{a_j}.
    \end{equation}
    Note that if $a_j>0$ for each $j$, then the upper bound in \cref{eq:localholomorphicitybounds} is $\infty$, and similarly, if $a_j<0$ for each $j$, then the lower bound is $-\infty$.
    
    \medskip
    
    We want to show that $I(\beta)$ has a meromorphic continuation to all of $\C_\beta$. The classical idea is to consider the following Bernstein--Sato type relations
    \begin{equation}
        \label{eq:bernstein1}
        \frac{\partial}{\partial w}|w|^{2(a\beta + b + 1)} = (a\beta + b + 1) \frac{|w|^{2(a\beta + b + 1)}}{w},\ \frac{\partial}{\partial \bar{w}}|w|^{2(a\beta + b + 1)} = (a\beta + b + 1) \frac{|w|^{2(a\beta + b + 1)}}{\bar{w}},
    \end{equation}
    which, for $\mathfrak{Re}(a \beta + b) > - 1$, are equalities of $L^1_{\mathrm{loc}}$-functions on $\C_w$. By an induction argument, it follows from \cref{eq:bernstein1} that
    \begin{equation}
        \label{eq:bernstein2}
        |w|^{2(a\beta + b)} = \prod_{j=1}^{m}(a\beta + b + j)^{-2} \frac{\partial^{2m}}{\partial w^m \partial \bar{w}^m} |w|^{2(a\beta + b + m)},
    \end{equation}
    for any $m \in \mathbb{N}$. Thus, for any $(m_1,\hdots,m_k)\in \mathbb{N}^k$, we have, by repeated application of \cref{eq:bernstein2}, that
    \[
        I(\beta) = h(\beta) \int_{\C^N} \frac{\partial^{2(m_1 + \cdots + m_k)} \big(|w_1|^{2(a_1\beta + b_1 + m_1)} \cdots |w_{k}|^{2(a_{k}\beta + b_{k} + m_k)}  \big)}{\partial w_1^{m_1} \partial \bar{w}_1^{m_1} \cdots \partial w_k^{m_k} \partial \bar{w}_k^{m_k}} \Psi(\beta) \,\d \widetilde{V}^{\otimes N},
    \]
    where
    \begin{equation}
        \label{eq:h}
        h(\beta) = \prod_{i=1}^k \prod_{j=1}^{m_i}(a_i\beta + b_i + j)^{-2}.
    \end{equation}
    Using integration by parts we have that
    \begin{equation}
        \label{eq:localintegral1}
        I(\beta) = h(\beta) \int_{\mathbb{C}^N} |w_1|^{2(a_1\beta + b_1 + m_1)} \cdots |w_{k}|^{2(a_{k}\beta + b_{k} + m_k)} \frac{\partial^{2(m_1 + \cdots + m_k)} \Psi(\beta)}{\partial w_1^{m_1} \partial \bar{w}_1^{m_1} \cdots  \partial \bar{w}_k^{m_k}} \,\d \widetilde{V}^{\otimes N}. 
    \end{equation}
    Note that the integral on the right-hand side of \cref{eq:localintegral1} converges and is holomorphic in the strip
    \[
        \max_{1\leq j \leq \ell} -\frac{1+b_j+m_j}{a_j} < \mathfrak{Re}\,\beta < \min_{\ell+1\leq j \leq k} -\frac{1+b_j+m_j}{a_j}.
    \]
    Moreover, for any $(m_1,\hdots,m_k)\in \mathbb{N}^k$, the function $h(\beta)$ given by \cref{eq:h} is a meromorphic function on $\C_\beta$, with poles given by
    \[
        \beta = - \frac{b_i + j}{a_i},\quad \text{for } j=1,\hdots,m_i,\quad i = 1,\hdots,k.
    \]
    Thus, since $(m_1,\hdots,m_k)$ can be chosen arbitrarily, it follows by letting each $m_j \to \infty$ that $I(\beta)$ has a meromorphic continuation to all of $\C_\beta$, with poles lying in a discrete subset of
    \[
        \Big(-\infty, \max_{1\leq j \leq \ell} -\frac{1+b_j}{a_j}\Big] \cup \Big[\min_{\ell+1\leq j \leq k} -\frac{1+b_j}{a_j}, \infty \Big).
    \]
    This implies that $\mathcal{Z}^\xi(\beta)$ has a meromorphic continuation to all of $\C_\beta$, with poles lying in a discrete subset of $(-\infty,\beta^-] \cup [\beta^+,\infty)$.

    \medskip
    \noindent
    \textbf{Step 2: The weak convergence of $\mu_\beta$.}
    
    From the previous step we know that, for any test function $\xi$, $\mathcal{Z}^\xi(\beta)$ extends to a meromorphic function on all of $\mathbb{C}_\beta$, hence we can consider its Laurent series expansion. For $\beta$ in a neighborhood of $\beta^-$ we have that
    \begin{equation}
        \label{eq: laurent series expansion}
        \mathcal{Z}^\xi(\beta) = \sum\limits_{j=-\kappa}^\infty (\beta-\beta^{-})^j \langle \nu_j(\beta^-),\xi\rangle,
    \end{equation}
    where $\langle \nu_j(\beta^-),\xi\rangle$ denotes the $j$\textsuperscript{th} order Laurent series coefficient. Here we take $\kappa$ to be the maximum over all test functions $\xi$ of the order of the pole of $\mathcal{Z}^\xi(\beta)$ at $\beta^-$. This is possible since, in view of Step 1 and \cref{eq:h} in particular, the order of the pole of $\mathcal{Z}^\xi$ at $\beta^-$, for any $\xi$, is at most $2N$. 
    As aforementioned, and as we will see below, 
    \[
        \nu_{-j}(\beta^-) \colon \xi \mapsto \langle\nu_{-j}(\beta^-),\xi\rangle
    \] 
    is a distribution on $(\mathbb{CP}^1)^N$ for each $j$. Thus, $\kappa$ can equivalently be defined as the unique positive integer such that $\nu_{-\kappa}(\beta^-)\not\equiv 0$ and $\nu_{-j}(\beta^-) \equiv 0$ for each $j<-\kappa$.

    If $\beta^- <\mathfrak{Re}\,\beta < \beta^+$, such that $\mathcal{Z}(\beta)$ is finite, we can consider the action $\langle \mu_\beta,\xi\rangle$ of the Gibbs measure $\mu_\beta$ on $\xi$. Recall that,
    \[
        \mathcal{Z}(\beta) = \mathcal{Z}^{\xi=1 }(\beta) = \sum\limits_{j=-\kappa}^\infty (\beta-\beta^{-})^j \langle \nu_j(\beta^-),1\rangle.
    \]
    Thus, in view of \cref{eq:actionofgibbs}, we have
    \begin{equation}
        \label{eq:mubeta at beta-}
        \begin{aligned}
            \langle \mu_\beta, \xi\rangle &= \frac{\mathcal{Z}^\xi(\beta)}{\mathcal{Z}(\beta)} \\
            &=\bigg(\sum\limits_{j=-\kappa}^\infty (\beta-\beta^-)^j \langle \nu_j(\beta^-),1\rangle\bigg)^{-1} \sum\limits_{j=-\kappa}^\infty (\beta-\beta^-)^j \langle \nu_j(\beta^-),\xi\rangle \\
            &= \frac{\langle \nu_{-\kappa}(\beta^-),\xi \rangle }{\langle \nu_{-\kappa}(\beta^-),1\rangle + \mathcal{O}(\beta-\beta^-)} + \frac{\sum_{j=-\kappa+1}^\infty (\beta-\beta^-)^{j+\kappa} \langle \nu_j(\beta^-),\xi\rangle }{\langle \nu_{-\kappa}(\beta^-),1\rangle + \mathcal{O}(\beta-\beta^-)},
        \end{aligned}
    \end{equation}
    which clearly converges to  $\langle \nu_{-\kappa}(\beta^-),\xi\rangle /\langle\nu_{-\kappa}(\beta^-),1\rangle$, as $\beta \to \beta^-$, as long as the denominator $\langle\nu_{-\kappa}(\beta^-),1\rangle \neq 0$. Given that $\nu_{-\kappa}(\beta^-)$ is a measure, then the denominator is nonzero in view of the definition of $\kappa$. Moreover, we get the weak convergence of $\mu_\beta$ towards the measure $\nu_{-\kappa}(\beta^-)/\langle \nu_{-\kappa}(\beta^-),1\rangle$. 

    \medskip

    \textit{Claim:} $\nu_{-\kappa}(\beta^-)$ is a measure. 
    
    To this end, $\langle \nu_{-\kappa}(\beta^-),\xi\rangle$ can be evaluated as
    \begin{equation}
        \label{eq: residue formula}
        \begin{aligned}
            \langle \nu_{-\kappa}(\beta^-),\xi\rangle &= \underset{\beta = \beta^-}{\mathrm{Res}} \big\{(\beta-\beta^-)^{\kappa - 1} \mathcal{Z}^\xi(\beta) \big\} \\
            &= \lim_{\beta \to \beta^-}(\beta-\beta^-)^{\kappa} \mathcal{Z}^\xi(\beta),
        \end{aligned}
    \end{equation}
    for any test function $\xi$. Recall from the previous step that, by pulling back to the Fulton--MacPherson compactification $(\mathbb{CP}^1)^{[N]}$ and introducing a partition of unity, we can equate $\mathcal{Z}^\xi(\beta)$ to a finite sum of integrals $I(\beta)$ of the form \Cref{eq:Zlocally}. For each such integral $I(\beta)$ we may assume, without loss of generality, that $a_1,\hdots,a_\ell > 0$ and that $a_{\ell+1},\hdots,a_{k} < 0$, and recall that $I(\beta)$ is defined and holomorphic for
    \[
        \max_{1\leq j \leq \ell} -\frac{1+b_j}{a_j} < \mathfrak{Re}\,\beta < \min_{\ell+1\leq j \leq k} -\frac{1+b_j}{a_j}.
    \]
    Let $1 \leq \ell_1 \leq \ell$ be the number of pairs $(a_j,b_j)$ that attain the maximum above and $1 \leq \ell_2 \leq k-\ell$ the number of pairs that attain the minimum. After a possible relabeling, we have that
    \begin{equation}
        \label{eq:localassumption1}
        \max_{1\leq j \leq \ell} -\frac{1+b_j}{a_j} = -\frac{1+b_1}{a_1} = \cdots = -\frac{1+b_{\ell_1}}{a_{\ell_1}} > \max_{\ell_1 + 1 \leq j \leq \ell} -\frac{1+b_j}{a_j},
    \end{equation}
    and
    \begin{equation}
        \label{eq:localassumption2}
        \min_{\ell+1\leq j \leq k } -\frac{1+b_j}{a_j} = -\frac{1+b_{\ell+1}}{a_{\ell+1}} = \cdots = -\frac{1+b_{\ell + \ell_2}}{a_{\ell + \ell_2}} < \min_{\ell + \ell_2 + 1 \leq j \leq k} -\frac{1+b_j}{a_j}.
    \end{equation}
    Let $\beta_{\mathrm{loc}}^-\coloneq - (1+b_1)/a_1$ and $\beta_{\mathrm{loc}}^+ = - (1+ b_{\ell+1})/a_{\ell+1}$. Furthermore, let
    \begin{equation}
        \label{eq:Gbeta}
        G_\beta(w) = \prod_{i=\ell_1 + 1}^{\ell}|w_{i}|^{2(a_{i}\beta + b_{i})} \prod_{j=\ell+\ell_2+1}^k|w_{j}|^{2(a_{j}\beta + b_{j})},
    \end{equation}
    which is locally integrable for $\beta \in [\beta_{\mathrm{loc}}^-,\beta_{\mathrm{loc}}^+]$ by \cref{eq:localassumption1,eq:localassumption2}. We have that
    \[
        I(\beta) = \int_{\C^N} \prod_{i=1}^{\ell_1}|w_i|^{2(a_i( \beta-\beta_{\mathrm{loc}}^-) - 1)} \prod_{j=\ell+1}^{\ell + \ell_2} |w_{j}|^{2(a_{j}( \beta-\beta_{\mathrm{loc}}^+) - 1)} G_\beta(w) \Psi(w) \,\d \widetilde{V}^{\otimes N},
    \]
    cf.\ \cref{eq:Zlocally}.

    There are now two cases: Either $\beta_{\mathrm{loc}}^- < \beta^-$, in which case $I$ is convergent at $\beta^-$, or $\beta_{\mathrm{loc}}^- = \beta^-$. In the first case the term $I$ gives no contribution to \cref{eq: residue formula}. Suppose that the latter is true. Write $\d \widetilde{V}^{\otimes N} = f(w) i^N\d w_1 \wedge \d \bar{w}_1 \wedge \cdots \wedge \d w_N \wedge \d \bar{w}_N$, where $f(w)$ is a smooth and strictly positive function. By repeated use of \cref{eq:bernstein1} we have that
    \begin{equation}
        \label{eq:Idistr}
        \begin{aligned}
        I &= \frac{(-1)^{\ell_1}}{a_1 \cdots a_{\ell_1}} (\beta-\beta^-)^{-\ell_1}\int_{\C^N} \bigg(\bigwedge_{i=1}^{\ell_1} \frac{\bar{\partial}|w_i|^{2a_i( \beta-\beta^-)}}{w_{i}} \wedge \d w_i \bigg) \wedge \prod_{j=\ell+1}^{\ell + \ell_2} |w_{j}|^{2(a_{j}( \beta-\beta_{\mathrm{loc}}^+) - 1)} \\
        &\qquad\qquad\qquad\qquad\qquad\,\,\,\, \times G_\beta(w) \Psi(w) f(w) {i^N} \d w_{\ell_1+1} \wedge \d \bar{w}_{\ell_1+1}\wedge \cdots \wedge \d w_N \wedge \d \bar{w}_N.
        \end{aligned}
    \end{equation}
    It is a standard result in residue theory that the distribution valued map
    \[
        \lambda \mapsto \frac{\bar{\partial}|w|^{2\lambda}}{w} \wedge \d w,
    \]
    is holomorphic in a neighborhood of $\mathfrak{Re}\,\lambda \geq 0$, and, moreover, that
    \[
        \frac{\bar{\partial}|w|^{2\lambda}}{w} \wedge \d w\bigg|_{\lambda = 0} = 2\pi i \delta_0(w),
    \]
    where $\delta_0(w)$ is the Dirac distribution. Similarly, the map
    \[
        (\lambda_1,\hdots,\lambda_{\ell_1}) \mapsto \frac{\bar{\partial}|w_1|^{2\lambda_1}}{w_1} \wedge \d w_1 \wedge \cdots \wedge \frac{\bar{\partial}|w_{\ell_1}|^{2\lambda_{\ell_1}}}{w_{\ell_1}} \wedge \d w_{\ell_1},
    \]
    which is just a $\ell_1$-fold tensor product of distributions on $\C$, is holomorphic in a neighborhood of the half-space $\{ (\lambda_1,\hdots,\lambda_{\ell_1})\in \C^{\ell_1} \colon \mathfrak{Re}\,\lambda_j \geq 0 \ \text{for } j=1,\hdots,\ell_1\}$, and 
    \begin{equation}
        \label{eq:product of deltas}
        \frac{\bar{\partial}|w_1|^{2\lambda_1}}{w_1} \wedge \d w_1 \wedge \cdots \wedge \frac{\bar{\partial}|w_{\ell_1}|^{2\lambda_{\ell_1}}}{w_{\ell_1}} \wedge \d w_{\ell_1}\bigg|_{\lambda_1 = \cdots = \lambda_{\ell_1} = 0} = (2\pi i)^{\ell_1} \delta_0(w_1) \wedge \cdots \wedge \delta_0(w_{\ell_1}).
    \end{equation}
    Thus, since 
    \[
        \prod_{j=\ell+1}^{\ell + \ell_2} |w_{j}|^{2(a_{j}( \beta-\beta_{\mathrm{loc}}^+) - 1)} G_\beta(w) \Psi(w)f(w)
    \]
    depends smoothly on $(w_1,\hdots,w_{\ell_1})$ (locally uniformly in $\beta$), see \cref{eq:Psibeta} and \cref{eq:Gbeta}, the integral on the right-hand side of \cref{eq:Idistr} is defined and holomorphic in a neighborhood of $\beta = \beta^-$. Moreover, $I$ has a pole of order at most $\ell_1$ at $\beta^-$ and, by \cref{eq:product of deltas}, we have that
    \begin{equation}
        \label{eq: local limit}
        \begin{aligned}
            \lim_{\beta \to \beta^-}(\beta-\beta^-)^{\ell_1} I(\beta) &= \frac{(2\pi)^{\ell_1}}{a_1 \cdots a_{\ell_1}} \int_{ \{w_1 = \cdots = w_{\ell_1} = 0 \}} \prod_{j=\ell+1}^{\ell+\ell_2} |w_{j}|^{2(a_{j}( \beta^- -\beta_{\mathrm{loc}}^+) - 1)} \times \\
            & \times G_{\beta^-}(w) \Psi(w) f(w) \, i^{N-\ell_1}\d w_{\ell_1+1} \wedge \d \bar{w}_{\ell_1+1}\wedge \cdots \wedge \d w_N \wedge \d \bar{w}_N. \\
        \end{aligned}
    \end{equation}
    From \cref{eq: local limit}, in view of \cref{eq:Psibeta,eq:Gbeta}, we see that if the test function $\xi$ is (strictly) positive, then $\lim_{\beta \to \beta^-}(\beta-\beta^-)^{\ell_1} I(\beta)$ is (strictly) positive. Indeed, the integrand on the right-hand side of \cref{eq: local limit} is generically positive and vanishing only on a set of positive codimension in $\{w_1=\cdots=w_{\ell_1}=0\}$. 

    \medskip
    
    Now, recall that $\kappa$ was defined as (the maximum over all test functions $\xi$ of) the order of the pole of $\mathcal{Z}^\xi(\beta)$ at $\beta^-$. Thus, there exists a test function $\xi$ such that
    \[
        \langle \nu_{-\kappa}(\beta^-),\xi\rangle = \lim_{\beta\to\beta^-}(\beta-\beta^-)^{\kappa}\mathcal{Z}^\xi(\beta) \neq 0.
    \]
    Moreover, $\kappa \geq \ell_1$, with $\ell_1$ as above. Clearly, $\lim_{\beta\to\beta^-}(\beta-\beta^-)^{\kappa}I(\beta) = 0$ unless $\ell_1 = \kappa$. Consequently, since \cref{eq: local limit} (with $\ell_1 = \kappa$) provides a local description of the non-vanishing contributions to $\langle \nu_{-\kappa}(\beta^-),\xi\rangle$, it follows that $\nu_{-\kappa}(\beta^-)$ is a non-trivial measure, proving the claim. Hence, $\langle \nu_{-\kappa}(\beta^-),1\rangle>0$ and that the weak limit in \cref{eq: weak limit of gibbs measures} from the statement of \Cref{thm:2} exists. Its action on a test function $\xi$ defined on $(\mathbb{CP}^1)^N$ is that of a positive measure $\widetilde{\mu}^-$, defined on the Fulton--MacPherson compactification $(\mathbb{CP}^1)^{[N]}$, acting on the pulled-back test function $\pi^* \xi$, that is,
    \[
        \langle \mu^-,\xi\rangle = \lim_{\beta\to\beta^-}\langle\mu_\beta,\xi\rangle = \langle \widetilde{\mu}^-,\pi^* \xi\rangle = \langle \pi_* \widetilde{\mu}^-,\xi \rangle,
    \]
    where the last equality is by definition of the pushforward of distributions with respect to proper maps.    

    \medskip
    \noindent
    \textbf{Step 3: The support of $\mu^-$}. 
    
    What is left is to determine the support of $\mu^-$. To this end, let us first understand the support of $\widetilde{\mu}^-$ on $(\mathbb{CP}^1)^{[N]}$. From \cref{eq: local limit} it is evident that $\widetilde{\mu}^-$ is supported on a union of certain intersections of irreducible hypersurfaces on $(\mathbb{CP}^1)^{[N]}$. More precisely, these hypersurfaces appear in either the support of $\pi^*(D_c)$ and, from \Cref{sec: FM}, any such hypersurface corresponds to a divisor $E_S$ for some $S\in G$. 
    
    We write
    \begin{equation}
        \label{eq: coefficients of E_S}
        \begin{aligned}
            &\pi^*(D_c) = \sum_{S\in G} a_S E_S, \\
            &K_{(\mathbb{CP}^1)^{[N]}/(\mathbb{CP}^1)^N} = \sum_{S\in G} b_S E_S.
        \end{aligned}
    \end{equation}
    Recall that $G^-$ is the subset of $G$ consisting of sets $S$ for which $a_S>0$ and $\beta^-=(1+b_S)/a_S$. From \cref{eq: local limit} we see that $\widetilde{\mu}^-$ is supported on the union of the maximal intersections of the divisors $E_S$ with $S\in G^-$. These intersection are maximal in the sense of having the maximal possible codimension. By \Cref{lemma: intersection pattern}, $E_{S_1},\dots,E_{S_k}$ intersect if and only if $S_1,\dots,S_k$ are pairwise nested; recall that $S_i$ and $S_j$ are nested if either $S_i \subseteq S_j$, $S_j \subseteq S_i$ or $S_i \cap S_j = \emptyset$. Thus, the support of $\widetilde{\mu}^-$ is given by
    \begin{equation*}
        \mathrm{supp}\,\widetilde\mu^- = \bigcup_{K\in \mathcal{N^-}}\bigcap_{S\in K} E_S.
    \end{equation*}
    We claim that
    \[
        \mathrm{supp}\, \mu^- = \bigcup_{K\in \mathcal{N^-}}\bigcap_{S\in K} W_S.
    \]
    To see this, first note that since $\pi$ is continuous, $\mathrm{supp}\, \pi_*(\widetilde{\mu}^-)=\overline{\pi(\mathrm{supp}\, \mu^-)}$. Thus the claim follows if we can show that
    \begin{equation}
        \pi\bigg(\bigcup_{K\in \mathcal{N^-}}\bigcap_{S\in K} E_S \bigg) = \bigcup_{K\in \mathcal{N^-}}\bigcap_{S\in K} \pi(E_S),
    \end{equation}
    since $\pi(E_S)=W_S$. We immediately have 
    \[
        \pi\bigg(\bigcup_{K\in \mathcal{N^-}}\bigcap_{S\in K} E_S\bigg)=\bigcup_{K\in \mathcal{N^-}}\pi\bigg(\bigcap_{S\in K} E_S\bigg)
    \]
    as well as the inclusion
    \[
        \pi\bigg(\bigcap_{S\in K} E_S \bigg) \subset \bigcap_{S\in K} \pi(E_S),\quad\forall K\in \mathcal{N^-}.
    \]
    To see that the converse inclusion holds, that is,
    \[
        \bigcap_{S\in K} \pi(E_S)\subset \pi\bigg(\bigcap_{S\in K} E_S\bigg),\quad\forall K\in \mathcal{N^-},
    \]
    pick $(p_1,\dots,p_N)\in \bigcap_{S\in K}W_S$. We must then find $q\in \bigcup_{S\in K}E_S$ such that $\pi(q)=(p_1,\dots,p_N)$. That such a $q$ exists is easiest to see via the description in \cite{FM} of $(\mathbb{CP}^1)^{[N]}$ using \emph{screens}, which is a convenient way to understand $(\mathbb{CP}^1)^{[N]}$ as a set. 
    
    If $p_1,\dots,p_N$ are all distinct, the fiber of $\pi$ over $(p_1,\dots,p_N)$ consists of a single point. Otherwise, an arbitrary point in the fiber is described by the following data. For each maximal collection of two or more indices $i_1,\dots,i_k$ such $p_{i_1} = \cdots = p_{i_k}$, one associates a screen to the set $\{i_1,\hdots,i_k\}$, that is, a tuple $q_1,\dots,q_k$ of $k$ points in the tangent space $T_{p_{i_1}}\mathbb{CP}^1$, not all coinciding; this data is only prescribed up to translation and homothety. Furthermore, whenever there are $\ell \geq 2$ indices $i_{k_1},\dots,i_{k_\ell}$ such that $q_{k_1}=\cdots =q_{k_\ell}$, one specifies an additional screen associated to $\{i_{k_1},\dots,i_{k_\ell}\}$, consisting of $\ell$ points in the tangent space $T_{q_{k_1}}(T_{p_{i_1}}\mathbb{CP}^1)$, not all coinciding, up to translation and homothety. 
    This iterative procedure continues until all new points introduced are distinct. The limiting nested collection of screens describes a unique point in the fiber over $(p_1,\hdots,p_N)$.
    
    In this picture, for $S=\{i_1,\dots,i_k\}$, the divisor $E_S$ corresponds to the locus of points in $(\mathbb{CP}^1)^{[N]}$ whose description includes a screen associated to the indices $i_1,\dots,i_k$. From this description, it is evident that the required choice of $q \in \bigcup_{S\subset K} E_S$ can be made by appropriately selecting the additional data defining the relevant screens. Such a choice is possible precisely because $K$ is a nest.
\end{proof}
\Cref{thm:3}, about the asymptotics of the partition function close to the critical inverse temperatures, is a consequence of the meromorphic continuation of $\mathcal{Z}(\beta)$ and its structure at $\beta^\pm$. 
\begin{proof}[Proof of \Cref{thm:3}]
    Consider the Laurent series expansion of $\mathcal{Z}^\xi(\beta)$ about $\beta=\beta^-$ in \cref{eq: laurent series expansion} with $\xi = 1$:
    \begin{equation*}
        \mathcal{Z}^1(\beta) = (\beta-\beta^-)^{-\kappa}\langle \nu_{-\kappa}(\beta^-),1\rangle + \mathcal{O}((\beta-\beta^-)^{-\kappa+1}),
    \end{equation*}
    where $\kappa$ is defined in Step 2 of the proof of \Cref{thm:2}. Thus,
    \begin{align*}
        -\frac{1}{N}\log(\mathcal{Z}(\beta)) &= \frac{\kappa}{N}\log(\beta-\beta^-)-\frac{1}{N}\log(\langle \nu_{-\kappa}(\beta^-),1\rangle +\mathcal{O}(\beta-\beta^-)) \\
        &= \frac{\kappa}{N}\log(\beta-\beta^-) - \frac{1}{N}\log(\langle\nu_{-\kappa}(\beta^-),1\rangle)-\frac{1}{N}\log\big(1+\mathcal{O}(\beta-\beta^-)\big) \\
        &= \frac{\kappa}{N}\log(\beta-\beta^-) + \mathcal{O}(1),
    \end{align*}
    as $\beta\to \beta^-$. The asymptotics as $\beta\to \beta^+$ follows analogously.

    It remains to show that $\kappa = \kappa^-$. Recall
    from the proof of \Cref{thm:2} that $\kappa$ equals the maximal number of simultaneously intersecting divisors $E$ in $(\mathbb{CP}^1)^{[N]}$, where $E$ is an irreducible component of $\mathrm{supp}(\pi^*(D))\cup\mathrm{supp}(K_{(\mathbb{CP}^1)^{[N]}/(\mathbb{CP}^1)^N})$ such that the coefficients $a$ and $b$ in front of $E$ in \cref{eq: rel can div W} and \cref{eq: order of Dc W}, respectively, satisfy $-(1+b)/a = \beta^-$. As noted above, these divisors are in bijection with subsets $S \subseteq \{1,\hdots,N\}$ that solve \Cref{prob: min}. Moreover, by \Cref{lemma: intersection pattern}, two or more of these divisors intersect if and only if the corresponding sets are nested, cf.\ the definition preceding \Cref{thm:2}. Consequently, $\kappa$ is the maximal cardinality of a collection of nested subsets solving \Cref{prob: min}, which is precisely the definition of $\kappa^-$. An analogous argument shows that the order of the pole of $\mathcal{Z}(\beta)$ at $\beta^+$ equals $\kappa^+$.  
\end{proof}
\section{Bounds on the critical inverse temperatures in terms of eigenvalues}\label{sec: eigenvalue bounds}
Recall that \Cref{prob: max,prob: min} are optimization problems over subsets $S$ of $\{1,\dots,N\}$ of size at least two. We can represent such a subset by a vector $\chi\in \{0,1\}^N$, where $\chi_i=1$ if $i\in S$ and $\chi_i=0$ otherwise. Let $C$ denote the symmetric $N\times N$ matrix with entries $c(i,j)$ for $i<j$, $c(j,i)$ for $i>j$ and zeros on the diagonal. With these identifications
\begin{equation}
    \label{eq: relation to Rayleigh ritz quotient}
    \frac{\sum_{i<j\in S}c(i,j)}{|S|-1}=\frac{1}{2}\frac{\chi^{\mathrm{T}}C\chi}{\chi^{\mathrm{T}}\chi-1}.
\end{equation}
The right-hand side is closely related to a Rayleigh--Ritz quotient, which allows us to obtain bounds on the critical temperatures in terms of the eigenvalues of $C$. 
\begin{proposition}
    \label{prop: bounds in terms of eigenvalues}
    Suppose that $\beta^+$ is finite. Then
    \begin{equation}
        \beta^+ \geq -1/\lambda_\mathrm{min}(C),
    \end{equation}
    where $\lambda_\mathrm{min}(C)$ is the smallest eigenvalue of $C$. Similarly, if $\beta^-$ is finite, then
    \begin{equation}
        \beta^- \leq -1/\lambda_\mathrm{max}(C),
    \end{equation}
    where $\lambda_\mathrm{max}(C)$ is the largest eigenvalue of $C$. 
\end{proposition}
\begin{proof}
    Recall that
    \begin{equation*}
        T^+\coloneq-\min_S \frac{\sum_{i,j \in S:i<j} c(i,j)}{|S|-1}.
    \end{equation*}
    Assume that $T^+>0$, that is, there exist $S\subset\{1,\dots,N\}$ such that $|S|\geq 2$ and $\sum_{i,j\in S:i<j} c(i,j)<0$. Using \cref{eq: relation to Rayleigh ritz quotient} we have
    \begin{equation*}
    \begin{aligned}
        -T^+ &= \min_{S\subset \{1,\dots,N\}: |S|\geq 2} \frac{\sum_{i<j\in S}c(i,j)}{|S|-1}\\ 
        &= \min_{\chi\in \{0,1\}^N:\chi^{\mathrm{T}}\chi\geq 2} \frac{1}{2}\frac{\chi^{\mathrm{T}}C\chi}{\chi^{\mathrm{T}}\chi-1} \\
        &\geq \min_{\chi\in \{0,1\}^N:\chi^{\mathrm{T}}\chi\geq 2} \frac{\chi^{\mathrm{T}}C\chi}{\chi^{\mathrm{T}}\chi} \\
        &\geq \min_{\chi\in \mathbb{R}^N} \frac{\chi^{\mathrm{T}}C\chi}{\chi^{\mathrm{T}}\chi} = \lambda_\mathrm{min}(C).
    \end{aligned}
    \end{equation*}
    where in the last step we use the min-max theorem for eigenvalues of a symmetric matrix and $\lambda_\mathrm{min}(C)$ is the smallest eigenvalue of $C$. By \Cref{thm:1} we obtain the bound $\beta^+\geq -1/\lambda_\mathrm{min}(C)$.
    
    A completely analogous argument shows that if $T^-<0$, then $\beta^-\leq 1/\lambda_\mathrm{max}(C)$ where $\lambda_\mathrm{max}(C)$ is the largest eigenvalue of $C$. 
\end{proof}
In the case $c(i,j)=k_i k_j$ for each $i\neq j$, letting $k_i$ be the components of a vector $k$ we have $C=kk^{\mathrm{T}} - \mathrm{diag}(k)^2$ where $\mathrm{diag}(k)$ is the $N\times N$ matrix with entries $k_i$ along the diagonal and zeros otherwise. This is quite useful in relation to \Cref{prop: bounds in terms of eigenvalues} since the only two eigenvalues of $kk^{\mathrm{T}}$ are $||k||^2$ and $0$. Using the Weil inequalities
\begin{equation*}
    \begin{aligned}
        &\lambda_\mathrm{max}(A+B)\leq \lambda_\mathrm{max}(A)+\lambda_\mathrm{max}(B), \\
        &\lambda_\mathrm{min}(A+B)\geq \lambda_\mathrm{min}(A)+\lambda_\mathrm{min}(B)
    \end{aligned}
\end{equation*}
for symmetric matrices $A$ and $B$ we prove \Cref{cor: explicit bounds in the charge case} from the introduction, Section \ref{cor: explicit bounds in the charge case}. 
\begin{proof}[Proof of \Cref{cor: explicit bounds in the charge case}]
    The first bound, \cref{eq: bound on beta plus for charges}, follows directly from \Cref{prop: bounds in terms of eigenvalues} and
    \begin{equation*}
        \lambda_\mathrm{min}(C) = \lambda_\mathrm{min}(kk^{\mathrm{T}}-\mathrm{diag}(k)^2) \geq \lambda_\mathrm{min}(kk^{\mathrm{T}})+\lambda_\mathrm{min}(-\mathrm{diag}(k)^2)=-\max_i k_i^2.
    \end{equation*}
    The second bound, \cref{eq: bound on beta minus for charges}, follows directly from \Cref{prop: bounds in terms of eigenvalues} and
    \begin{equation*}
        \lambda_\mathrm{max}(C) = \lambda_\mathrm{max}(kk^{\mathrm{T}}-\mathrm{diag}(k)^2) \leq \lambda_\mathrm{max}(kk^{\mathrm{T}})+\lambda_\mathrm{max}(-\mathrm{diag}(k)^2) = ||k||^2-\min_i k_i^2.
    \end{equation*}
\end{proof}
\section{Random coupling}
\label{sec: random coupling}
In general, Problems \ref{prob: max} and \ref{prob: min} are not tractable to solve. However, an interesting special case to consider is when $c(i,j)$ are random variables. We will look at the following two cases: When the particles have independent random couplings, and when they have independent random charges.

\subsection{Random coupling}

Let us first assume the couplings $c(i,j)$ are i.i.d.\ standard normal random variables. In this setting, the bounds in \Cref{prop: bounds in terms of eigenvalues} lead to stochastic bounds using the well known distributions of the maximal eigenvalue of the \textit{Gaussian Orthogonal Ensemble (GOE)}. 
\begin{proposition}
\label{prop: independent random coupling}
    Let $c(i,j)$ be independent, normally distributed random variables with mean $0$ and variance $1/N$, for each $1\leq i<j\leq N$. Then
    \begin{equation}
        T^+ \leq 2 
    \end{equation}
    almost surely, and
    \begin{equation}
        T^- \geq -2
    \end{equation}
    almost surely. 
\end{proposition}

\begin{proof}
    Let $C$ be, as above, the random $N\times N$ symmetric matrix with zeros on the diagonal and off-diagonal entries given by the random variables $c(i,j)$ for $i<j$ and $i>j$. Define $\widetilde{C}\coloneq C+2D$ where $D$ is a diagonal matrix with i.i.d.\ standard normal entries $d_i$. Note that
    \[
        \mathbb{P}(T^+> 0)=1-0.5^{N(N-1)/2}\to 1\quad \text{as } N\to \infty,
    \]
    since $T^+>0$ as soon as $c(i,j)<0$ for some pair $(i,j)$. Assuming $T^+>0$, so that \Cref{prop: bounds in terms of eigenvalues} applies, the \textit{Weil inequality} yields
    \begin{equation*}
        T^+ \leq -\lambda_\mathrm{min}(C) \leq -\lambda_\mathrm{min}(\widetilde{C})-2\lambda_\mathrm{min}(D).
    \end{equation*}
    The first term on the right-hand side, $-\lambda_\mathrm{min}(\widetilde{C})$, corresponds to the largest eigenvalue of a GOE random matrix. To leading order, its expectation equals 2 and its variance scales as $AN^{-1/3}$, where $A > 0$ \cite{TrWi}. 
    
    The second term, $-2\lambda_{\min}(D)$, can be written as $\frac{2}{N}\max(d_i)$. This quantity is well studied in extreme value theory, and by, e.g., \cite[Example 1.1.7]{HaFe}, to leading order, both its expectation and variance grow as $A'\sqrt{\log(N)}$ for some positive constant $A'$. Consequently, 
    \begin{equation*}
        T^+ \leq 2 \quad \text{almost surely}.
    \end{equation*}
    The analysis for $T^-$ is carried out analogously. 
\end{proof}
Note that, by definition, $T^+\leq 2$ does not exclude that $T^+<0$ which implies $\beta^+=\infty$. However, as stated in the above proof, the probability of this event tends to zero as $N\to \infty$. 

\subsection{Random charges}
Another natural random model is to consider random, independent charges instead of random, independent coupling parameters. This setup has been studied, for instance, in \cite{KiWa}. It appears particularly natural in the context of the Onsager model of turbulence, where, to the best of the authors' knowledge, there is no canonical ansatz for the vorticities of the vortex particles. In this case, \Cref{cor: explicit bounds in the charge case} yields corresponding stochastic bounds. 
\begin{proposition}
    \label{prop: independent random charges}
    Let $k_i$, for $i=1,\dots,N$, be i.i.d.\ standard normal random variables and let $c(i,j)=k_ik_j$. Then we have the stochastic bounds $T^+\leq C_+$ and $T^-\geq C_-$, where
    \[
        2(C_+-\log(N) +\log\log(N)/2+\log(\Gamma(1/2))\to \Lambda,
    \]
    with $\Lambda$ standard Gumbel distributed, and where 
    \[
        \frac{1}{\sqrt{2N}}(C_- - \sqrt{N}/2)\to \Theta,
    \]
    where $\Theta$ is standard normally distributed, as $N\to \infty$.
\end{proposition}
\begin{proof}
    Since $\mathbb{P}(k_i = 0) = 0$ for each $i=1,\hdots,N$, $T^+>0$ almost surely. By \Cref{cor: explicit bounds in the charge case}, conditional on this event,
    \[
        T^+ \leq \max_i k_i^2.
    \]
    Each random variable $k_i^2$ follows a $\Gamma(1/2,2)$ distribution. Standard results from extreme value theory then imply that
    \[
        2\big(\max_i k_i^2-\log(N) +\log\log(N)/2+\log(\Gamma(1/2) \big)\to \Lambda \quad \text{as } N \to \infty,
    \]
    where $\Lambda$ is a standard Gumbel random variable.

    For the negative critical temperature, note that $T^-<0$ as long as there is at least one positive $k_i$ and at least one negative $k_i$, respectively. Thus,
    \[
        \mathbb{P}(T^-<0) = 1-2\cdot 0.5^N\to 1 \quad \text{as } N\to \infty.
    \]
    By \Cref{cor: explicit bounds in the charge case}, conditional on $T^-<0$,
    \[
        -T^-\leq \sum_i k_i^2-\min_i k_i^2 \leq \sum_i k_i^2
    \]
    and, by the central limit theorem,
    \[
        \frac{1}{\sqrt{2N}}\Big(\sum_i k_i^2-\sqrt{N}/2 \Big)\to \Theta \quad \text{as } N \to \infty,
    \]
    where $\Theta$ is standard normally distributed. 
\end{proof}

\section{Applications}
\label{sec: examples}
\subsection{Positive temperature}
\label{sec: positive temperature examples}

In this section we prove \Cref{thm: two-component log gas}. We begin by considering the two-component plasma defined by the coupling matrix
\begin{equation}
    \label{eq: c for general two-component plasma, v2}
    c(i,j) = \begin{cases}
        Z^2, 1\leq i,j \leq N_1\\
        1, N_1< i,j \leq N\\
        -Z, 1\leq i \leq N_1, N_1<j\leq N\\
        -Z, N_1< i \leq N, 1\leq j\leq N_1,
    \end{cases}
\end{equation}
for some fixed real number $Z \in [1,\infty)$.
\begin{lemma}
    \label{lemma: computation of beta for two component gas}
    Let $c(i,j)$ be as in \eqref{eq: c for general two-component plasma, v2}. Then
    \[
        \beta^+=\bigg(\max_{|S|\geq 2} \frac{\sum_{i<j\in S} c(i,j)}{-|S|+1}\bigg)^{-1} = 1/Z.
    \]
    Moreover, the maximum is achieved precisely for any $S \subset \{1,\hdots,N\}$ of size $|S|=2$ such that $|S_1|=|S_2|=1$, where $S_1 = S\cap\{1,\hdots,N_1\}$ and $S_2 = S\cap \{N_1 + 1,\hdots,N\}$. 
\end{lemma}
To prove \Cref{lemma: computation of beta for two component gas}, we will make use of the following technical lemma. 
\begin{lemma}
    \label{lemma: technical lemma}
    For $Z \in [1,\infty)$ and integers $a,b$ of the form $a=2k-1$, $b=2\ell-1$ where $(k,\ell) \in \mathbb{Z}_{\geq 0}^2 \setminus \{(0,0)\}$, we have either that
    \begin{equation}
        \label{eq:ineq1}
       |Z a-b|\geq Z-(a+b-1) 
    \end{equation}
    or that 
    \begin{equation}
        \label{eq:ineq2}
        a+b\geq Z+1.
    \end{equation}
\end{lemma}
\begin{proof}
    Assume first $Z a-b\geq 0$. If $b\geq 1$, then $a \geq 1$. If $b=-1$, then by assumption $a\geq 1$. In either case we have that
    \begin{equation*}
        Z a-b\geq Z - (a+b-1).
    \end{equation*}
    Consider now the case when $Z a-b <0$. If $a=-1$, then, by assumption, $b \geq 1$, whence
    \begin{equation*}
        Z  a-b = -Z -b +(a+1) \leq - Z -b +(a+1) + 2(b-1) = -Z + (a+b-1),
    \end{equation*}
    and \cref{eq:ineq1} follows.
    
    If $a\geq 1$, then $Z  \leq Z  a <b\leq a+b-1$. So $Z  +1 \leq a+b$. 
\end{proof}
\begin{proof}[Proof of \Cref{lemma: computation of beta for two component gas}]
    Starting from \eqref{eq: c for general two-component plasma, v2} we have
    \begin{equation*}
        \begin{aligned}
            \sum_{i<j\in S} c(i,j) &= Z^2|S_1|(|S_1|-1)/2+|S_2|(|S_2|-1)/2-Z|S_1||S_2| \\ 
            \phantom{\sum_{i<j\in S} c(i,j)} &= (Z|S_1|-|S_2|)^2/2-Z^2|S_1|/2-|S_2|/2 \\
            \phantom{\sum_{i<j\in S} c(i,j)} &= (Z|S_1|-|S_2|-(Z-1)/2)^2/2-Z|S|/2-(Z-1)^2/8 \\
            \phantom{\sum_{i<j\in S} c(i,j)} &= (Z(2|S_1|-1)-(2|S_2|-1))^2/8-Z|S|/2-(Z-1)^2/8.
        \end{aligned}
    \end{equation*}
    We note that for $S$ such that $|S_1|=|S_2|=1$,
    \begin{equation*}
        \frac{-\sum_{i<j\in S} c(i,j)}{|S|-1} = Z.
    \end{equation*}
    It remains to show that, for arbitrary $S$ of size at least $2$,
    \begin{equation*}
        \frac{-\sum_{i<j\in S} c(i,j)}{|S|-1} \leq Z,
    \end{equation*}
    with equality if and only if $|S_1|=|S_2|= 1$. To this end, let $a=2|S_1|-1$ and $b=2|S_2|-1$. Note that $|S|=|S_1|+|S_2|=(a+b)/2+1$. Note also that $Z$, $a$ and $b$ satisfy the conditions of \Cref{lemma: technical lemma}. We divide the argument into two cases.
    %
        
    \medskip 
    \noindent
    \textbf{Case 1:} Suppose $a+b \geq Z + 1$. We have that
    \begin{align*}
        \frac{-\sum_{i<j\in S} c(i,j)}{|S|-1} &= \frac{Z|S|/2+(Z-1)^2/8 -(Z(2|S_1|-1)-(2|S_2|-1))^2/8}{|S|-1} \\
        \phantom{\frac{-\sum_{i<j\in S} c(i,j)}{|S|-1}}&\leq \frac{Z|S|/2+(Z-1)^2/8}{|S|-1} \\
        \phantom{\frac{-\sum_{i<j\in S} c(i,j)}{|S|-1}}&= \frac{Z((a+b)/2+1)/2+(Z-1)^2/8}{(a+b)/2} \\
        \phantom{\frac{-\sum_{i<j\in S} c(i,j)}{|S|-1}}&= \frac{Z}{2} + \frac{Z + (Z-1)^2/4}{a+b} \\
        \phantom{\frac{-\sum_{i<j\in S} c(i,j)}{|S|-1}}&=\frac{Z}{2} + \frac{(Z + 1)^2}{4(a+b)}.
    \end{align*}
    By the assumption that $a+b \geq Z + 1$ we thus find that
    \[
        \frac{-\sum_{i<j\in S} c(i,j)}{|S|-1} \leq \frac{Z}{2} + \frac{Z + 1}{4} \leq Z,
    \]
    since $Z \geq 1$. We see that for equality to hold in the last inequality, we need $Z = 1$. But then the assumption $a + b \geq Z + 1$ becomes $2|S| \geq 3$, which can never be sharp.
    %

    \medskip
    \noindent
    \textbf{Case 2:} Suppose that $a+b < Z + 1$. By \Cref{lemma: technical lemma} we then have that $|Z a - b| \geq Z - (a+b) + 1$. It follows that
    \begin{align*}
        \frac{-\sum_{i<j\in S} c(i,j)}{|S|-1} &= \frac{Z|S|/2+(Z-1)^2/8 -(Za - b)^2/8}{|S|-1} \\
        \phantom{\frac{-\sum_{i<j\in S} c(i,j)}{|S|-1}}&\leq \frac{Z|S|/2+(Z-1)^2/8 -(Z-(a+b)+1)^2/8}{|S|-1} \\
        \phantom{\frac{-\sum_{i<j\in S} c(i,j)}{|S|-1}}&= \frac{Z|S|/2+(Z-1)^2/8 -(Z-2|S|+3)^2/8}{|S|-1}, \\
    \end{align*}
    where in the last step we used that $(a+b) = 2|S| - 2$.
    Rewriting the numerator on the right-hand side,
    \[
        Z|S|/2+(Z-1)^2/8 -(Z-2|S|+3)^2/8 = (|S|-1)(Z - (|S|-2)/2),
    \]
    we find that
    \begin{align*}
        \frac{-\sum_{i<j\in S} c(i,j)}{|S|-1} &\leq Z - \frac{|S|-2}{2} \leq Z.
    \end{align*}
    For equality to hold in the last inequality, we must have that $|S|=2$. For equality to hold overall, we must have that
    \begin{align*}
        |Z a - b| &= Z - (a+b) + 1 \\
        \iff |Z(2|S_1|-1) - (2|S_2|-1)| &= Z - 1.
    \end{align*}
    If $|S_1| = 2$ and $S_2 = \emptyset$, the above becomes
    \[
        3Z + 1 = Z - 1,
    \]
    which has no solution in $[1,\infty)$. If $S_1 = \emptyset$ and $|S_2| = 2$, we obtain
    \[
        Z + 3 = Z - 1,
    \]
    which has no solution for any $Z$. This leaves only one remaining possibility, namely $|S_1| = |S_2| = 1$, which we already saw attains the maximum. This concludes the proof.
    %
\end{proof}
We are now ready to prove \Cref{thm: two-component log gas}.

\begin{proof}[Proof of \Cref{thm: two-component log gas}]
    For a general two-component plasma model, defined by \cref{eq: general planar two-component plasma} with $Z_1,Z_2 \in (0,\infty)$, we may assume, without loss of generality, that $Z_2 \leq Z_1$. Then, by making the change of variable $\widetilde{\beta} = Z_2^2 \beta$, and letting $Z \coloneq Z_1 / Z_2 \in [1,\infty)$, we reduce to the setting of \Cref{lemma: computation of beta for two component gas}, whence $\widetilde{\beta}^+ = 1/Z = Z_2/Z_1$. Changing back to $\beta$, we immediately find that $\beta^+ = 1/Z_1 Z_2$.

    From the proof of \Cref{lemma: computation of beta for two component gas} we also see that $G^-$ consists of sets $S$ with exactly two elements, $i$ and $j$, where $i< N_1$ and $j> N_1+1$, corresponding to one particle of each charge. Thus, $\mathcal{N}^-$ consists of collections $K$ of subsets of $\{1,\dots,N\}$ of size $N_1$ where each subset in $K$ has two elements, one at most $N_1$ and one at least $N_1+1$. It follows that $\kappa^+=N_1=N/(1+Z_1/Z_2)$ where the last equality is due to overall charge balance. 
\end{proof}
\subsection{Negative temperature}
\label{sec: negative temperature results}
Recall the general coupling defined by
\begin{equation}
    \label{eq: onsager model}
    c(i,j) = k_i k_j,
\end{equation}
for nonzero real numbers $k_i$. Without loss of generality, we assume that $k_i>0$ for $i=1,\dots,N_1$ and $k_i<0$ for $i=N_1+1,\dots,N_2$ where $N_1+N_2=N$. To be able to solve \Cref{prob: min}, we put a strong assumption on the charges $k_i$, namely that 
\begin{equation}\label{eq: cond 1 onsager}
    \max_{i\colon k_i>0} k_i< \frac{3}{2}\min_{i\colon k_i>0} k_i,
\end{equation}
and that
\begin{equation}\label{eq: cond 2 onsager}
    \max_{i\colon k_i<0} |k_i|< \frac{3}{2}\min_{i\colon k_i<0} |k_i|.
\end{equation}
We will also assume that $N>2$ to avoid the degenerate case $c(1,2)<0$.
\begin{corollary}
    \label{cor: onsager model}
    Let $c(i,j)$ be given by \cref{eq: onsager model}, and assume that the conditions \eqref{eq: cond 1 onsager} and \eqref{eq: cond 2 onsager} hold. Then
    \begin{equation}
        \beta^- = \max\left[ \frac{-N_1+1}{\sum_{i<j\leq N_1}k_ik_j} , \frac{-N_2+1}{\sum_{N_1<i<j}k_ik_j}\right]
    \end{equation}
    and
    \[
        \mathrm{supp}\,\mu^- = 
        \begin{cases}
            \{p_1=\dots=p_{N_1}\} &\text{if } \dfrac{-N_1+1}{\sum_{i<j\leq N_1}k_ik_j} > \dfrac{-N_2+1}{\sum_{N_1<i<j}k_ik_j}, \\[1.2em]
            \{p_{N_1+1}=\dots=p_{N}\} &\text{if }\dfrac{-N_1+1}{\sum_{i<j\leq N_1}k_ik_j} < \dfrac{-N_2+1}{\sum_{N_1<i<j}k_ik_j}.
        \end{cases}
    \]
    Furthermore, if 
    \[
        \frac{-N_1+1}{\sum_{i<j\leq N_1}k_ik_j} \neq \frac{-N_2+1}{\sum_{N_1<i<j}k_ik_j},
    \]
    then
    \begin{equation}
        -\frac{1}{N}\log\mathcal{Z}(\beta) = \frac{1}{N}\log(\beta-\beta^-)+\mathcal{O}(1).
    \end{equation}
\end{corollary}
\begin{proof}
Recall from \Cref{thm:1}, in view of \cref{eq: onsager model}, that $\beta^-$ is finite and given by
\begin{equation}
    \label{eq: max for onsager model in negative temperature}
    \beta^- = -\bigg( \max_{S} \frac{\sum_{i<j\in S} k_i k_j}{|S|-1} \bigg)^{-1}. 
\end{equation}
We claim that the maximum is achieved for a subset $S$ for which $\ell \in S$ implies $k_\ell \sum_{i\in S \setminus\{\ell\}} k_i > 0$. To seek a contradiction, assume that $S$ realizes the maximum in \cref{eq: max for onsager model in negative temperature} and that $\ell \in S$ is such that $k_\ell \sum_{i\in S} k_j \leq 0$. We can assume that $\sum_{i<j\in S}k_ik_j >0$. Estimate
\begin{equation*}
    (|S|-2)\sum_{i<j\in S} k_i k_j < (|S|-1)\bigg(\sum_{i<j\in S\setminus \{\ell\}}k_ik_j +k_\ell \sum_{i \in S\setminus\{\ell\}} k_i \bigg) \leq (|S|-1)\sum_{i<j\in S\setminus \{\ell\}} k_ik_j,
\end{equation*}
and after rearranging
\begin{equation*}
    \frac{\sum_{i<j\in S\setminus \{\ell\}} k_ik_j}{|S\setminus \{\ell\}|-1}>\frac{\sum_{i<j\in S} k_ik_j}{|S|-1},
\end{equation*}
which contradicts that $S$ realizes the maximum in \cref{eq: max for onsager model in negative temperature}, proving the claim. This implies that the maximum is in fact achieved for a $S$ for which all $k_i$'s with $i\in S$ have the same sign. To see this, assume, without loss of generality, that $\sum_{i\in S}k_i \geq 0$, and that $k_1 < 0$. Then, clearly, $k_1 \sum_{i\in S \setminus\{1\}}k_i < 0$.

Assuming further that $\max_{i:k_i>0} |k_i|< \frac{3}{2} \min_{i:k_i>0} |k_i|$ and similarly for the case of negative $k_i$'s, we claim that the maximum is achieved precisely for a subset $S$ with the maximum number of elements possible for which all $k_i's$ with $i\in S$ have the same sign. To see this, assume that $S$ realizes the maximum but is not maximal in the above sense and let $\ell \notin S$ be such that $k_\ell$ has the same sign as the $k_i$'s with $i\in S$. Write 
\[
    |S|\sum_{i<j\in S} k_i k_j = (|S|-1)\sum_{i<j\in S} k_ik_j +\sum_{i<j\in S} k_ik_j,
\]
and estimate
\begin{align*}
    \sum_{i<j\in S} k_ik_j &= \frac{1}{2} \bigg( \sum_{i,j\in S}k_i k_j - \sum_{i\in S}k_i^2 \bigg)\\
    &\leq \frac{1}{2} \bigg( |S|\max_{j: k_j>0}k_j \sum_{i\in S}k_i - \min_{j: k_j>0}k_j  \bigg) \sum_{i\in S}k_i \\
    &\leq \frac{1}{2}\Big( |S|\max_{j: k_j>0}k_j  - \min_{j: k_j>0}k_j \Big) \sum_{i\in S}k_i \\
    &\leq \frac{1}{2}\Big( \frac{3}{2} |S| - 1 \Big) \min_{j: k_j>0}k_j \sum_{i\in S}k_i \\
    &\leq  (|S|-1) \min_{j: k_j>0}k_j \sum_{i \in S}k_i,
\end{align*}
where, in the second-to-last step, we used our assumption, and in the last step the fact that $|S|\geq 2$. Thus, since $k_\ell \geq \min_{j: k_j>0}k_j$, we have
\[
    |S| \sum_{i<j\in S} k_i k_j \leq (|S|-1) \Big(  \sum_{i<j\in S} k_i k_j + k_\ell \sum_{i\in S} k_i \Big) = (|S|-1) \sum_{i<j\in S \cup \{\ell\}} k_i k_j.
\]
Rearranging yields
\begin{equation*}
    \frac{\sum_{i<j\in S\cup \{\ell\}} k_ik_j}{|S\cup \{\ell\}|-1} > \frac{\sum_{i<j\in S} k_ik_j }{|S|-1},
\end{equation*}
and therefore
\begin{equation*}
    \beta^- = \max\left[ \frac{-N_1+1}{\sum_{i<j\leq N_1}k_ik_j} , \frac{-N_2+1}{\sum_{N_1<i<j}k_ik_j}\right].
\end{equation*}
\end{proof}
\begin{example}
    \label{ex: no bound on rel sizes of vorticities}
    If the bounds \eqref{eq: cond 1 onsager} and \eqref{eq: cond 2 onsager} on the relative sizes of the $k_i$ are not assumed, then the conclusion of \Cref{cor: onsager model} need not hold. For instance, consider the case $N=3$ with $k_1=a, k_2=a$, $k_3=1$. Then 
    \begin{equation*}
        \frac{\sum_{i<j\in S}k_ik_j}{|S|-1} = \begin{cases}
            \frac{a^2+2a}{2} &\text{ if }S=\{1,2,3\}\\
            a^2 &\text{ if }S=\{1,2\}\\
            a &\text{ if }S=\{1,3\}\text{ or }S=\{2,3\}.
        \end{cases}
    \end{equation*}
    Hence, the maximum is attained uniquely for $S=\{1,2\}$ when $a$ is sufficiently large. 
\end{example}
\begin{remark}
    \label{rmk:gibbs}
    If $k_i = k \in \mathbb{R}$ for each $i=1,\hdots,N$, then the bounds \eqref{eq: cond 1 onsager} and \eqref{eq: cond 2 onsager} hold. By \Cref{cor: onsager model}, the support of $\mu^-$ is then given by $\{p_1=\dots=p_N\}$. In other words, total collapse occurs as the negative critical inverse temperature is approached. 
\end{remark}

\begin{remark}
    For $k_i=\sqrt{2/(N-1)}$, \Cref{cor: onsager model} yields $\beta^-=-1+1/N$, in agreement with \cite{Be} and \cite{Fu}. In this case, $\lim_{N\to \infty} -\beta^-$ is precisely the $\gamma$-invariant of the Fano manifold $\mathbb{CP}^1$. Moreover, $\lim_{N\to \infty} -\beta^-=1$ reflects the fact that $\mathbb{CP}^1$ is \textit{Gibbs semistable}, but not \textit{Gibbs stable}, see \cite{Be} and \cite{Fu}. 
\end{remark}

\section{Relations to other optimization problems}
\label{sec: other problems}

\subsection{Arboricity of a graph}

As in previous sections, let $C$ be the symmetric $N\times N$ matrix with zeros on the diagonal and entries $c(i,j)$ for $i<j$ and $c(j,i)$ for $i>j$. When $C$ is the adjacency matrix of an undirected graph $G$, the quantity $\lceil-T^- \rceil$, where $T^-$ is defined in \Cref{prob: min}, coincides with the \textit{arboricity} of $G$. The arboricity of an undirected graph $G$ is defined as the minimal number of forests that partition the edge set of $G$. By the main result in \cite{NW}, the arboricity of a graph is given by 
\begin{equation}
    \label{eq: arboricity}
    \max_{H\subseteq G} \bigg\lceil \frac{|E(H)|}{|V(H)|-1} \bigg\rceil
\end{equation}
where the maximum is taken over all subgraphs $H$ of the graph $G$, and where $E(H)$ and $V(H)$ denote the sets of edges and vertices of $H$, respectively. The quantity
\begin{equation*}
    \label{eq: fractional arboricity}
    \max_{H\subseteq G} \frac{|E(H)|}{|V(H)|-1}
\end{equation*}
is often referred to as the \textit{fractional arboricity} in the literature. 

In the more general setting of matroids, determining the arboricity \cref{eq: arboricity} of a graph generalizes to finding the minimal number of independent subsets into which a matroid can be partitioned. This is known as the \textit{matroid partitioning problem}, and a formula analogous to \cref{eq: arboricity} was established in \cite{Ed}, and there is also a version related to \cref{eq: fractional arboricity} in \cite{ScUl}.

Finally, note that when $C$ is an adjacency matrix, the quantity $-\beta^+$ coincides with the log canonical threshold of a certain ideal associated to a reduced hyperplane arrangement, as explained in Section 2. For a general reduced hyperplane arrangement, the log canonical threshold was computed in \cite{Mu}. The formula for the log canonical threshold established in \cite{Mu} appears to be related to the matroid partitioning problem \cite{ScUl} in a similar way as described above in the case of graphs.

\subsection{Random couplings and the Sherrington--Kirkpatrick model}

Let us now go back to the case of random couplings in \Cref{sec: random coupling}. Curiously, this problem resembles a problem coming from the statistical study of spin glasses, namely the so called Sherrington--Kirkpatrick model. In this model, the space of states is given by a spin vector $\sigma_i\in\{-1,1\}^N$ and the Hamiltonian is given by
\begin{equation}
    \label{eq: SK hamiltonian}
    H(\sigma)=-\sum_{i,j} J_{i j}\sigma_i\sigma_j +h\sum_i \sigma_i. 
\end{equation}
where $h$ is some fixed external field and $J_{i j}$, for $1\leq i,j\leq N$, are random interaction terms, typically assumed to be i.i.d.\ Gaussian. The main problem in this setting is to understand the ground state energy,
\begin{equation}
    \label{eq: ground state energy}
    H_\mathrm{min} \coloneq \min_{\sigma \in \{-1,1\}^N} H(\sigma),
\end{equation}
for example in terms of its expectation over the random couplings. As is standard in statistical mechanics, this can be reduced to computing the limit of the free energy
\[
    F(\beta)=\frac{1}{N}\mathbb{E}[\log Z_N(\beta)],
\]
where $Z_N(\beta)$ is the partition function
\[
    Z_N(\beta)=\sum_{\sigma\in \{-1,1\}^N} \exp(-\beta H(\sigma)),
\] 
as $\beta\to \infty$. A solution was found with heuristic methods by Parisi in \cite{Par}, the famous ``Parisi solution'', which was later rigorously proven by Talagrand in \cite{Ta}.

To see a connection with the main problem of this paper, we will assume that the spins take values in $\{0,1\}$ instead. That is, put $\sigma=\chi$ in \cref{eq: SK hamiltonian} and \cref{eq: ground state energy} where $\chi$ is a vector in $\{0,1\}^N$ as in \Cref{sec: eigenvalue bounds}. Note that the main result of \cite{Ta} includes this case as well. Let $C$ denote the symmetric $N\times N$ matrix with entries $c(i,j)$ for $i<j$, $c(j,i)$ for $i>j$ and zeros on the diagonal. Given a solution to \Cref{prob: min}, that is, a vector $\chi'\in \{0,1\}^N$ solving the optimization problem
\begin{equation}
    \label{eq: our opt problem}
    \max_{\substack{\chi\in\{0,1\}^N: \\ \chi^{\mathrm{T}}\chi \geq 2}} \frac{\chi^{\mathrm{T}}C\chi}{\chi^{\mathrm{T}}\chi-1} = -T^-,
\end{equation}
then $\chi'$ is also a solution to the optimization problem
\begin{equation}
    \label{eq: ground state energy 2}
    \min_{\substack{\chi\in\{0,1\}^N: \\ \chi^{\mathrm{T}}\chi \geq 2}} \bigg( -\chi^\mathrm{T} C \chi -T^- \sum_{i=1}^N \chi_i \bigg) =-T^-, 
\end{equation}
where we note that $\sum_{i=1}^N \chi_i = \chi^{\mathrm{T}} \chi$. Conversely, a solution to \cref{eq: ground state energy 2} for which $h=-T^-$, or any solution for which $H_\mathrm{min}=h$, is a solution to \cref{eq: our opt problem} as well. A similar statement is true also regarding \Cref{prob: max} and $T^+$. Thus, it seems that the Parisi solution of the Sherrington--Kirkpatrick model with spins taking values in $\{0,1\}$ might be useful to understand the critical temperature of the Log gas with random couplings. Indeed, for fixed deterministic couplings and for a carefully chosen external field $h$, the ground state energy of the Sherrington--Kirkpatrick model is precisely the negative of the critical temperature of the corresponding Log gas. Although, knowing the expectation of $H_\mathrm{min}$ does not directly yield any knowledge about the expectation of $T^-$. Moreover, whereas in the Sherrington--Kirkpatrick model, the coupling $c(i,j)$ are assumed i.i.d.\ Gaussian, the assumption in \cite{KiWa} seems more natural, where the charges $k_i$ are assumed i.i.d.\ Gaussian. This version of the Sherrington--Kirkpatrick model does not seem to have been considered previously.

\end{document}